%% file: SublinearCliqueFull.tex
\documentclass[11pt,twoside]{article}

\usepackage{parskip}

\usepackage{fullpage}
\usepackage{comment}

\usepackage{epsf}
\usepackage{fancyheadings}
\usepackage{graphics}
\usepackage{graphicx}
\usepackage{psfrag}

\usepackage{color}

\usepackage{bbold}

\usepackage{amsthm}
\usepackage{amsfonts}
\usepackage{amsmath}
\usepackage{amssymb}

\usepackage{graphicx,epsfig,graphics,epsf}
\usepackage{float,amsmath,amsfonts,amssymb,amsthm}
\usepackage{mathtools}
\usepackage{subcaption}

\usepackage{amsmath,amssymb}
\usepackage{amsthm}
\usepackage{mathtools}
\usepackage[noend]{algorithmic}
\usepackage[ruled,vlined]{algorithm2e}
\usepackage{url}
\usepackage{fullpage}
\usepackage{makeidx}
\usepackage{enumerate}
\usepackage[top=1in, bottom=1in, left=1in, right=1in]{geometry}
\usepackage{graphicx,float,psfrag,epsfig,caption}
\usepackage[usenames,dvipsnames,svgnames,table]{xcolor}
\definecolor{darkgreen}{rgb}{0.0,0,0.9}

\theoremstyle{plain}

\newtheorem{innercustomgeneric}{\customgenericname}
\providecommand{\customgenericname}{}
\newcommand{\newcustomtheorem}[2]{%
  \newenvironment{#1}[1]
  {%
   \renewcommand\customgenericname{#2}%
   \renewcommand\theinnercustomgeneric{##1}%
   \innercustomgeneric
  }
  {\endinnercustomgeneric}
}

\newcustomtheorem{customthm}{Theorem}
\newcustomtheorem{customlemma}{Lemma}
\newcustomtheorem{customproposition}{Proposition}
\newcustomtheorem{customcorollary}{Corollary}

\newlength{\widebarargwidth}
\newlength{\widebarargheight}
\newlength{\widebarargdepth}

\makeatletter
\long\def\@makecaption#1#2{
        \vskip 0.8ex
        \setbox\@tempboxa\hbox{\small {\bf #1:} #2}
        \parindent 1.5em  
        \dimen0=\hsize
        \advance\dimen0 by -3em
        \ifdim \wd\@tempboxa >\dimen0
                \hbox to \hsize{
                        \parindent 0em
                        \hfil 
                        \parbox{\dimen0}{\def\baselinestretch{0.96}\small
                                {\bf #1.} #2
                                } 
                        \hfil}
        \else \hbox to \hsize{\hfil \box\@tempboxa \hfil}
        \fi
        }
\makeatother

\long\def\comment#1{}

\newcommand{\Prob}{\ensuremath{{\mathbb{P}}}}

\DeclareSymbolFont{rsfs}{U}{rsfs}{m}{n}
\DeclareSymbolFontAlphabet{\mathscrsfs}{rsfs}

\numberwithin{equation}{section}

\newtheoremstyle{myexample} 
    {\topsep}                    
    {\topsep}                    
    {\rm }                   
    {}                           
    {\bf }                   
    {.}                          
    {.5em}                       
    {}  

\newtheoremstyle{myremark} 
    {\topsep}                    
    {\topsep}                    
    {\rm}                        
    {}                           
    {\bf}                        
    {.}                          
    {.5em}                       
    {}  

\newtheorem{claim}{Claim}[section]
\newtheorem{lemma}[claim]{Lemma}
\newtheorem{fact}[claim]{Fact}

\newtheorem{conjecture}[claim]{Conjecture}
\newtheorem{theorem}{Theorem}
\newtheorem{proposition}[claim]{Proposition}

\newtheorem{definition}[claim]{Definition}

\theoremstyle{myremark}
\newtheorem{remark}{Remark}[section]

\usepackage{pgfplots}
\pgfplotsset{compat=1.5}
\usepgfplotslibrary{fillbetween}

\DeclarePairedDelimiterX{\inp}[2]{\langle}{\rangle}{#1, #2}

\DeclareMathOperator{\pr}{\mathbb{P}}

\newcommand{\G}{{\sf G}}
\newcommand{\PCC}{\emph{Planted Clique Conjecture}}
\newcommand{\A}{\mathcal{A}}
\newcommand{\PCD}{{\sf{PC_D}}(n,k)}
\newcommand{\PCR}{{\sf{PC_R}}(n,k)}
\newcommand{\Binnhalf}{{\sf Bin}\left(n,\frac{1}{2}\right)}
\newcommand{\Binnkhalf}{{\sf Bin}\left(n-k,\frac{1}{2}\right)}

\usepackage{bbm}
\usepackage[inline]{enumitem}
\usepackage{enumitem}

\usepackage[pagebackref,letterpaper=true,colorlinks=true,pdfpagemode=none,citecolor=OliveGreen,linkcolor=BrickRed,urlcolor=BrickRed,pdfstartview=FitH]{hyperref}

\title{Finding Planted Cliques in Sublinear Time}
\author{Jay Mardia\thanks{Department of Electrical Engineering, Stanford
    University. \{jmardia, asi, kabirc\}@stanford.edu} \and Hilal Asi\footnotemark[1] \and Kabir Aladin Chandrasekher\footnotemark[1]}
\begin{document}

\date{}
%
\maketitle
\begin{abstract}
We study the planted clique problem in which a clique of size $k$ is planted in
an Erd\H{o}s-R\'enyi graph $G(n, \frac{1}{2})$ and one is interested
in recovering this planted clique. It is widely believed that it exhibits a statistical-computational gap when computational efficiency is equated with the existence of polynomial time algorithms. We study this problem under a more fine-grained computational lens and consider the following two questions- 
\begin{enumerate}
	\item \emph{Do there exist sublinear time algorithms for recovering the planted
		clique?}
	\item \emph{What is the smallest running time any algorithm can hope to have?}
\end{enumerate}

We show that because of a well known clique-completion property, very elementary sublinear time recovery algorithms do indeed exist for clique sizes $k=\omega(\sqrt{n})$. This points to a qualitatively stronger statistical-computational gap. The planted clique recovery problem can be solved without even looking at most of the input above the $\Theta(\sqrt{n})$ threshold and cannot be solved by any efficient algorithm below it.

A running time lower bound for 
the recovery problem follows easily from the results of~\cite{racz2019finding}, and this implies our recovery algorithms are optimal whenever 
$k = \Omega(n^{\frac{2}{3}})$. However, for $k=o(n^{\frac{2}{3}})$ there is a gap between our algorithmic upper bound and the information-theoretic lower bound implied by~\cite{racz2019finding}. 

With some caveats, we show stronger \textit{detection} lower bounds based on the \textit{Planted Clique Conjecture} for a natural but restricted class of algorithms. The key idea is to relate very fast sublinear time algorithms for detecting large planted cliques to polynomial time algorithms for detecting small planted cliques.
\end{abstract}

\setcounter{tocdepth}{2}

\section{Introduction}
\label{sec:Introduction}
\input{Introduction.tex}



\section{Algorithms}
\label{sec:algorithms}
\input{algorithms.tex}

\section{Auxiliary Lemmas}
\label{sec:auxlem}
\input{auxlemmas.tex}

\section*{Acknowledgments}
The authors would like to thank Amir Abboud for helpful discussions about fine grained complexity that improved the presentation of our results. KAC would like to thank Kannan Ramchandran for asking a question that helped lead to this work.



\addcontentsline{toc}{section}{References}
\bibliographystyle{alpha}
\bibliography{ref}




\end{document}

%% file: Introduction.tex
The planted clique problem, in which a clique of size $k$ is planted in
an Erd\H{o}s-R\'enyi graph $G(n, \frac{1}{2})$ has been well studied over the past two decades. It
has emerged as a fruitful testbed to understand the misalignment between statistical and computational notions of efficiency. Much research has led to the belief that efficient polynomial time algorithms only exist for clique sizes $k$ larger than $\Omega(\sqrt{n})$. Above this threshold, previous work \cite{feige2010finding,dekel2014finding,deshpande2015} has shown (nearly) linear time 
algorithms that run in time $\widetilde{O}(n^2)$. However, no lower bounds were proven to establish
the optimality of these algorithms.

One reason to stop at linear time algorithms is that sublinear time algorithms almost inherently entail a non-zero probability of failure, since they must randomly choose what small portion of the input to see. However, for statistical problems, it is already true that no algorithm whatsoever can succeed with zero error probability. It is then intriguing to ask if 
it is possible to recover the planted clique more efficiently 
by looking at only a small subset of the graph. We investigate the following two questions:

\begin{enumerate}
\item \emph{Do there exist sublinear time algorithms for recovering the planted
clique?}
\item \emph{What is the smallest running time any algorithm can hope to have?}
\end{enumerate}

We provide partial answers 
to these questions by developing sublinear time algorithms for the
planted clique problem and establishing some evidence---based on the 
\PCC---of their optimality.

\subsection{Our Contribution}

\textbf{Algorithms:} 

Algorithm~\ref{alg:subsampledcounting} runs in time $\widetilde{O}(n^{3/2})$ and recovers the clique with high
probability for $k = \Theta(\sqrt{n\log{n}})$. Algorithm~\ref{alg:subsample+khdac} is very simple extension\footnote{While it is possible to formulate Algorithm~\ref{alg:subsampledcounting} and Algorithm~\ref{alg:subsample+khdac} as a single unified algorithm, we prefer to keep them separate for pedagogical clarity.} and shows that for even
larger clique sizes there is an
$\widetilde{O}\left(\left(n/k\right)^3 + n\right)$ algorithm for clique
recovery. 
Finally, Algorithm~\ref{alg:subsample+filter} runs in time
${O}\left(n^2/\left(\frac{k^2}{n}\exp{\left(\frac{k^2}{24n}\right)}\right)\right)$ and recovers the planted clique when $\omega(\sqrt{n}) = k = o(\sqrt{n\log{n}})$. Given the widespread belief (which goes by the name \PCC) that no polynomial time algorithm can recover the planted clique if $k = O(n^{\frac{1}{2}-\delta})$ for any constant $\delta >0$, we certainly do not expect sublinear time algorithms to work in that regime.

Thus our work builds towards the idea that the planted clique problem can either be solved without even looking at the entire graph, or it needs more than a polynomial amount of time. In a qualitative sense, this is a much bigger ‘gap’ than that suggested by lumping all polynomial time algorithms together under one umbrella.

We stress that our algorithms are extremely straightforward, and do not involve much technical innovation. The key idea that enables sublinear time recovery is that of clique completion, a concept already widely known in the literature. While none of the existing clique completion routines run in sublinear time, it is quite easy to create one. Our algorithmic contribution lies in realising that clique completion allows sublinear time algorithms, rather than any new technical ideas.

\textbf{Impossibility Results:}

We begin our investigation of the second question by observing that the results 
of \cite{racz2019finding} imply that any recovery algorithm 
requires time at least $\Omega(\frac{n^2}{k^2} + n)$. This has the following slightly surprising consequence. Once the planted clique is of size at least $k = \Omega\left( n^{\frac{2}{3}}\right) $, we get an optimal recovery runtime of $\widetilde{\Theta}(n)$. Increasing the size of the planted clique further does not make the recovery problem easier. In fact, detection versions of our algorithm do not require the additive $\tilde{O}(n)$ running time, and run in time $\widetilde{O}\left(\frac{n^3}{k^3}\right)$. 

This means that for large $k$, there is indeed a fine-grained computational complexity separation between the detection and recovery problems, even though we know from \cite{alon2007testing} that recovery can be reduced to detection.

The lower bound techniques of \cite{racz2019finding} are purely information theoretic, and it can be seen that such techniques will not be able to prove stronger lower bounds, of the form $\Omega(\frac{n^3}{k^3})$ that we might hope for given our algorithmic results. 
To circumvent this, we aim to show stronger lower bounds using widely accepted average case computational hardness assumptions.
The most natural such assumption in this scenario is, evidently, the \PCC. 
We aim to build on hardness of the planted clique problem for small cliques (as codified by the \PCC) to show lower bounds for algorithms that recover large planted cliques. Our goal is to investigate the (however rough) notion that the hardness of the planted clique problem in all regimes is due to the same reason. We note that the rest of our lower bounds work for the easier detection version of the planted clique problem; 
these bounds imply lower bounds for the recovery problem.

We make some progress towards this goal, but with a caveat. The reductions we show use a slightly different notion of a planted clique problem, which we call ${\sf{iidPC}_D}$. In this model, each vertex is included in the clique independently with probability $\frac{k}{n}$. In the vanilla planted clique (${\sf{PC}_D}$) problem, the clique is a uniformly random subset of $k$ vertices. Our observation here is that reducing between planted clique problems with different clique sizes is very easy to do in the ${\sf{iidPC}_D}$ world, even though it is not so easy in the ${\sf{PC}_D}$ world. It is not uncommon in the literature that impossibility results for the planted clique problem use slightly differing formulations. For example \cite{feldman2017statistical} use an iid bipartite version of the problem to show impossibility results for statistical query algorithms. We discuss this issue further in the subsequent sections.

We emphasize that our starting hardness assumption can be based on either of the two formulations of planted  clique, since it is easy to show (Remark~\ref{rem:pcc-to-iidpcc}) that the \PCC~for the standard variant implies the analogous conjecture for ${\sf{iidPC}_D}$. For now, we state our results while pretending that they formally hold for the vanilla planted clique problem. Unfortunately, we can only show this connection for a restricted family of algorithms\footnote{Our algorithms for planted clique recovery are technically not in this class, but detection versions of these algorithms \textit{are}. We elaborate more on this in later sections.}. While we would like to show such a connection for broader classes of algorithms (and provide some intuition for why such a fact might be true), we are currently unable to. This is a weakness of our result.

We show that
assuming the \PCC, no \textit{non-adaptive}\footnote{Definition~\ref{def:nonadaptive-alg}} \textit{rectangular}\footnote{Definition~\ref{def:rect-alg}} algorithm can detect the existence of a planted clique of size $k = \widetilde{\Omega}(\sqrt{n})$ in time
$O\left(n^{3-\delta}/k^3 \right)$ for any constant $\delta > 0$. We have thus transformed a computational hardness assumption that distinguishes between polynomial and superpolynomial time algorithms into a result that distinguishes between more fine-grained (in fact sublinear) running times.

A careful look at the reduction used to prove this raises the possibility of a sublinear query complexity version of a statistical-computational gap. Using computationally inefficient algorithms, planted cliques can be detected with just $\widetilde{O}(\frac{n^2}{k^2})$ \cite{racz2019finding} queries to the adjacency matrix. However, our result suggests that efficient polynomial time algorithms require, at the very least, $\widetilde{\Omega}(\frac{n^3}{k^3})$ queries to the adjacency matrix. While we can only prove this for the restricted class of non-adaptive rectangular algorithms, it is plausible that such a result holds for much broader classes of algorithms. This sort of gap is in contrast to the more common notion of statistical-computational gap in terms of the size of the planted clique.

In the other direction, we show that for planted cliques of size $k = \Theta( \log n \sqrt{n})$, any detection runtime lower bound of the form $\omega(n)$ gives a non-trivial $\omega(n^2)$ runtime lower bound for detecting planted cliques of size $k = 3 \log n$ (i.e. near the information theoretic threshold below which detection is information theoretically impossible). While this is nowhere near as spectacular as claiming that no polynomial time algorithms can exist, it \textit{is} a conditional super-linear lower bound.

As discussed, our results have several caveats and restrictions. We hope that these are just the first steps in showing that the non-existence of fast sublinear time algorithms for detecting large cliques is related to the hardness of detecting small cliques.

\subsection{Open problems}

\begin{enumerate}
\item The running times of our algorithms for planted clique sizes  just above and just below 
$\Theta(\sqrt{n \log n})$ are dramatically different. For $k = \Omega(\sqrt{n \log n})$, we can recover the clique in time $\widetilde{O}(n + (\frac{n}{k})^{3}) = \widetilde{O}(n^{\frac{3}{2}})$. For $k = o(\sqrt{n \log n}) $, our algorithms are not even `truly sublinear', by which we mean that they run slower than $\Omega(n^{2-\delta})$ for any constant $\delta > 0$. Is there  some threshold phenomenon at clique size $k = \Theta(\sqrt{n \log n})$ with such different behavior above and below it, or are there faster algorithms for smaller cliques? Both positive and negative answers to the following become very interesting.
\begin{quote}
	\emph{Does there exist an algorithm which runs in time $O(n^{2 -
			\delta})$ for some constant $\delta > 0$ which can recover planted cliques of size $k =
		o(\sqrt{n \log n})$?}
\end{quote}

\item Detection versions of our algorithms are non-adaptive and rectangular. To complement this, we have shown that the $\PCC$ implies non-existence of non-adaptive rectangular algorithms that reliably solve the detection problem and run much faster than our algorithms. This leads to wondering about the power of general algorithms with an adaptive sampling strategy.
\begin{quote}
	\emph{Does there exist an adaptive and/or non-rectangular algorithm which runs in time $O(n^{\frac{3}{2} -
			\delta})$ for some constant $\delta > 0$ and reliably detects planted cliques of size $k =
		\Theta(\sqrt{n \log n})$?}
\end{quote}

\item 

To show strong lower bounds, we have relied on the most natural computational hardness assumption for this setting, namely the \PCC. However, it is plausible that other assumptions might be relevant too. Can we gain evidence for the non-existence of fast sublinear time algorithms that solve the planted clique problem using other computational hardness assumptions?
\end{enumerate}

\section{Related work}
As far as the authors are aware, the planted clique problem was first studied
in~\cite{jerrum1992large} in which Jerrum studied Markov chain Monte Carlo
methods and showed that the metropolis process cannot find cliques of size
$O(\sqrt{n})$.  It is known that just above the information theoretic threshold, $k = 2\log{n}$, there is a unique largest clique with high probability and
the brute force algorithm will successfully find the clique.  This lies in stark
contrast to where polynomial-time algorithms begin to work.  The first polynomial time algorithm was provided
in~\cite{kuvcera1995expected} although shown only to work above the degree counting
threshold $k = \Omega(\sqrt{n\log{n}})$. Several algorithms were later shown to work
for $k = \Omega(\sqrt{n})$, starting with the spectral algorithm from \cite{alon1998finding}, and including an algorithm that is based on semidefinite programming from \cite{feige2000finding}.	In fact, a line of work including more sophisticated degree counting
algorithms~\cite{feige2010finding,dekel2014finding} and approximate message passing~\cite{deshpande2015} has shown that cliques of size larger than $\Omega(\sqrt{n})$ can be found in nearly linear ($\widetilde{O}(n^2)$) time. To the best of our knowledge, no sublinear time algorithm has been proposed so far.

On the flip side, it is widely believed that no polynomial time algorithm can solve the planted clique problem for clique size significantly smaller than $O(\sqrt{n})$. Evidence for this fact has mounted up in recent years, and comes from showing that restricted classes of algorithms can not beat this bound. $\Theta(\sqrt{n})$ was shown to be a barrier for the powerful sum-of-squares
hierarchy~\cite{meka2015sum,deshpande2015improved,hopkins2018integrality,barak2019nearly} and for statistical query
algorithms~\cite{feldman2017statistical}. Hardness against circuit classes for a related problem was shown in \cite{rossman2008constant,rossman2014monotone}. This body of work has provided evidence
for a so-called statistical-computational gap.  We refer the reader
interested in statistical-computational gaps in planted problems to the
papers~\cite{wu2018statistical, bandeira2018notes,
	gamarnik2019landscape}.

As a result, a number of works have used this conjectured hardness, the \PCC~(or close variants) to show average-case hardness results for various
problems~\cite{alon2007testing,arora2010computational,berthet2013complexity,koiran2014hidden,ma2015computational,wang2016average, brennan2018reducibility,shah2019feeling,manurangsi2020strongish}.  It has
additionally been used as a cryptographic primitive~\cite{applebaum2010public}.  We follow in these footsteps by using the \PCC~as our main hardness assumption to prove lower bounds for sublinear time algorithms. A key difference here is that instead of using an assumption that talks about that gap between polynomial and superpolynomial time algorithms to obtain another such gap, we use it to show a fine-grained (in fact sublinear) hardness result that distinguishes between different polynomial running times.

In recent years, there have indeed been reductions of this form. Such connections have resulted in the burgeoning field of fine-grained complexity (see \cite{williams2018some} for a nice survey), including  the study of a fine-grained understanding of clique problems \cite{abboud2018if}.  Recently, \cite{ball2017average} explored one of the first fine-grained average case complexity results by using the random self-reducibility of low-degree polynomials to turn worst case fine-grained hardness results into average case results. However, it should be noted here that our techniques are quite different to these works. We rely on the fact that when we look at only a small fraction of our input, the loss of information makes it look indistinguishable from a problem that should not have any polynomial time algorithm. Further, \cite{goldreich2018counting,boix2021average,dalirrooyfard2020new,goldreich2020counting} studied the relation between worst-case and average-case hardness of clique problems.

More recently, \cite{feige2020finding,racz2019finding,alweiss2021subgraph,feige2021tight} have considered the problem of finding cliques in random graphs where the cost of the algorithm is the number of queries it makes to the adjacency matrix of the input graph. This is similar to our framework in that this quantity, the number of queries, plays a central role in both our algorithms and impossibility results. However, both of these works only bound the number of queries but allow unbounded computation time, while we only allow a sublinear amount of computation. In fact, our interest in the number of queries is simply a byproduct of this requirement. Following initial dissemination of this work, \cite{huleihel2021random} provided an algorithm that generalizes one of our results to work for the more general planted dense subgraph problem.

\section{Our techniques}

\subsection{Algorithms}
\label{sec:techn-algs}
All of our algorithms build on a simple
idea: once an algorithm has found slightly more than $\log n$ (say $2 \log n$)
clique vertices, we can efficiently (and with high probability of success) find all the other planted clique vertices.

Subroutines of this form are not new, and are well-known in the planted
clique literature~\cite[Lemma 2.9]{dekel2014finding}. However, we need our subroutine to run in time $\widetilde{O}(n)$ (which is optimal) and
without knowledge of the planted clique size. The clique completion lemma in \cite{dekel2014finding} both needs $k$ to be specified, and runs in time $\Omega(k^2)$, which could be $\widetilde{\Omega}(n)$, and so is unsuitable for our purposes.

We circumvent this and create a clique completion subroutine with the desired properties.  
Given $2 \log n$ vertices that are in the planted clique, every other planted clique vertex is connected to all these $2 \log n$ vertices and with high probability very few non-clique vertices are connected to all these $2 \log n$ initial clique vertices. Thus we can restrict our attention to only those vertices which are connected to all $2 \log n$ initial clique vertices. Call this set $V'$. There might, however, be some false positives in $V'$, because our input to this subroutine might be any (adversarially chosen\footnote{The property that our subroutine works correctly with an arbitrary input subset of the planted clique rather than only a uniformly random one is crucial when we use it in Algorithm~\ref{alg:subsample+filter} to solve the planted clique problem for cliques of size $o(\sqrt{n \log n})$.}) planted clique subset of the ${k \choose 2 \log n}$ possibilities\footnote{To gain intuition for why there are false positives, we can try a union bound argument to show that there are no false positives. For a given subset of the planted clique of size $2 \log n$, except with probability at most $\frac{1}{n^2}$, there will be no false positives. However, we then need to union bound over all ${k \choose 2 \log n}$ possible clique subsets, at which point such an anlysis breaks down.}. We simply select a random subset $S'_C$ of size $2 \log n$ from $V'$. With high probability, this subset $S'_C$ will contain only clique vertices, and then we run the same ``common neigbour'' procedure on this small subset. Note that now $S'_C$ is not just `some' (possibly adversarially chosen) subset of the planted clique, but is in fact a uniformly random subset. We can then utilise this randomness to show that with high probability no non-clique vertex is connected to all $2 \log n$ elements in $S'_C$. More details can be found in Algorithm~\ref{alg:cliquecompletion} and we prove a formal version of the following statement in Lemmas~\ref{lem:clique-completion-runtime} and \ref{lem:clique-completion-correctness}.

\begin{quote}
\textsc{Clique-Completion}
(Algorithm~\ref{alg:cliquecompletion}) takes a subset of the clique of
size $2\log{n}$ and, in running time $O(n\log{n})$, returns the planted clique
with high probability (as long as $k  = \omega \left( \log^2 n \right) $).

\end{quote}

\subsubsection{An $\widetilde{O}(n^{3/2})$ algorithm for
	finding cliques of size $k = \Theta(\sqrt{n \log n})$}
\begin{quote}
Theorem~\ref{thm:subsampledcounting} in Section~\ref{subsec:subsampledcounting}
formalises an algorithm
\textsc{Keep-High-Degree-And-Complete} (Algorithm~\ref{alg:subsampledcounting})
which runs in time $\widetilde{O}(n^{3/2})$ and finds the planted clique with high
probability of success as long as the clique size $k \geq C\sqrt{n\log{n}}$ for
a large enough constant $C$.
\end{quote}

The algorithm is elementary and follows from the same simple observations that led \cite{kuvcera1995expected} to give the first polynomial time algorithm for planted clique at the same threshold $k \geq C\sqrt{n\log{n}}$.
\begin{enumerate}
	\item The degree of each non-clique vertex is distributed as $\Binnhalf$.
	There are at most $n$ non-clique vertices, and with high probability the
	maximum of $n$ (possibly dependent) $\Binnhalf$ random variables is at
	most $\frac{n}{2} + c\sqrt{n\log{n}}$ for some constant $c>0$.
	\item With high probability, the degree of all the clique vertices will be
	larger than $\frac{n-k}{2} + k - c\sqrt{n \log n} = \frac{n}{2} +
	\frac{C}{2}\sqrt{n\log{n}} - c\sqrt{n \log n}$.
\end{enumerate}
If we choose $C$ large enough, simply computing the degree of a vertex (which takes time $O(n)$) lets us decide if the vertex is in the clique or not.
Since $k$ out of the $n$ vertices are in the clique, if we randomly
sample slightly more than $\frac{n}{k}$ vertices from $V$, we will get
at least $2 \log n$ clique vertices, and can identify them by computing the degree of all the vertices we have sampled, which takes time $\widetilde{O}\left( \frac{n^2}{k}\right)  = \widetilde{O}\left( n^{\frac{3}{2}}\right) $. Then we simply use the \textsc{Clique-Completion} subroutine to find the entire clique in a further $O(n \log n)$ time.\\

\begin{remark}\label{rem:robust}
	This algorithm works even if we only have an underestimate of the true planted clique size. This is because the algorithm only uses the clique size implicitly, when deciding how many vertices to sample. If we underestimate the clique size, we will only sample more vertices than necessary. This will increase the runtime, but will not affect the correctness of the output.
\end{remark}

\subsubsection{An $\widetilde{O}\left(\left(n/k\right)^3 + n\right)$ algorithm for
	finding cliques of size $k = \Omega(\sqrt{n \log n})$}
Theorem~\ref{thm:subsampledcounting} give a runtime of $\widetilde{O}\left(\frac{n^2}{k} +n \right)$ for $k = \Omega(\sqrt{n \log n})$. However, \textsc{Keep-High-Degree-And-Complete} was tailored for $k = \Theta\left(\sqrt{n\log n} \right)$. 
\begin{quote}
Theorem~\ref{thm:subsample+khdac} in Section~\ref{subsec:n3k3} shows that for larger $k$, we can improve the runtime to $\widetilde{O}\left( \left( \frac{n}{k}\right)^3 + n \right)$ using \textsc{Subsample-And-KHDAC}, Algorithm~\ref{alg:subsample+khdac}.
\end{quote}

Let $p<1$ be some parameter. If we select a random subset of $pn$  vertices (which can be done in time $O(pn)$), we expect to have $pn$ vertices of which $pk$ are planted clique vertices. If $pk = \Omega(\sqrt{pn \log{(pn)}	})$, we can run \textsc{Keep-High-Degree-And-Complete} on this smaller problem instance in time $\widetilde{O}\left(\left( pn\right) ^{\frac{3}{2}} \right)$ to recover $pk$ planted clique vertices. We can then just run \textsc{Clique-Completion} on a subset of them in time $O(n \log n)$. 
Observe that we need $p \approx \frac{n}{k^2}$ for this to work, 
resulting in a runtime of $\widetilde{O}\left( \left( \frac{n}{k}\right)^3 + n \right)$.

\subsubsection{An ${O}\left(n^2/\left(\frac{k^2}{n}\exp{\left(\frac{k^2}{24n}\right)}\right)\right)$ algorithm for finding cliques of size $\omega(\sqrt{n}) = k = o(\sqrt{n \log n})$}

\begin{quote}
Theorem~\ref{thm:subsample+filter} analyses \textsc{Subsample-And-Filter} (Algorithm~\ref{alg:subsample+filter}) and shows that even when $\omega(\sqrt{n})  = k = o(\sqrt{n \log n})$, degrees do help solve the planted clique recovery problem in sublinear time ${O}\left(n^2/\left(\frac{k^2}{n}\exp{\left(\frac{k^2}{24n}\right)}\right)\right)$.
\end{quote}
However, our algorithm  is not `truly sublinear'. That is, it does not have running time $O(n^{2-\epsilon})$ for any constant $\epsilon>0$ . We leave the question of devising a `truly sublinear' algorithm for finding the planted clique when $k = o(\sqrt{n \log n})$ as a compelling open problem.

The reason degree counting (as in \cite{kuvcera1995expected} and Theorem~\ref{thm:subsampledcounting}) works for finding planted cliques of size $\Omega(\sqrt{n \log n})$ is because there exists a clear separation between the degree of clique vertices and non-clique vertices. If we see a vertex that has degree close to $\frac{n+k}{2}$, we know it is in the clique, and if the degree is much lesser than $\frac{n+k}{2}$ (even if it is much larger than $\frac{n}{2}$), we know it is not in the clique. The situation changes when $\omega(\sqrt{n}) = k = o(\sqrt{n \log n})$. A vertex with degree close to (or even much larger than) $\frac{n+k}{2}$ may be a non-clique vertex.

However, all is not lost. Given a clique vertex, its degree is very likely to be close to its expectation of $\frac{n+k}{2}$. On the other hand, given a non-clique vertex, its degree is much less likely to be close to $\frac{n+k}{2}$, even though this likelihood is not as small as in the case of $k = \Omega(\sqrt{n \log n})$. This suggests that we filter out vertices based on this closeness criterion.

We subsample an i.i.d $p$ fraction of the vertices (this can be done in time $O(n)$), and then compute the degree of each of these (approximately) $pn$ vertices. This takes time $O(pn^2)$. We then throw away all vertices that are not within $O(\sqrt{n})$ of $\frac{n+k}{2}$. This will boost the ratio of clique to non-clique vertices because of the discussion above. If we choose $p$ large enough so that at the end of this process we get $n'$ vertices in all, out of which $k'$ are planted clique vertices, and $k' = \Omega(\sqrt{n'})$, then we can use any existing linear time solver to finding planted cliques on this smaller problem instance. This takes time $O({n'}^2) = O(p^2n^2) = O(pn^2)$. The specific linear time solver we use is the \textsc{Low-Degree-Removal} algorithm of \cite{feige2010finding}. This is because we do not have access to $k'$, and unlike the algorithms of \cite{dekel2014finding,deshpande2015}, \textsc{Low-Degree-Removal} has the nice property that it does not need the size of the planted clique as input. We can then use \textsc{Clique-Completion} as a final step to find all the clique vertices in the original problem. As we see during the analysis, it suffices to take $p = {\Theta}\left(\frac{n}{k^2}\exp\left(-\frac{k^2}{24n} \right)  \right)$.

For technical reasons, Algorithm~\ref{alg:subsample+filter} is actually slightly different from the sketch described above. We first split the vertices of the input graph into two disjoint sets $V_1$ and $V_2$ of equal size $n/2$. We then subsample vertices from $V_1$, and use their $V_2$-degree to filter them. By $V_2$-degree we mean that we estimate their degree by only counting the number of edges from a vertex in $V_1$ to all the vertices in $V_2$. The advantage now is that when we take our filtered vertices (which are a subset of $V_1$) and consider the subgraph induced by them, we have not seen any of the edges in this subgraph. Thus we can use the randomness of these edges to argue that this subgraph is an instance of the planted clique problem and can invoke \cite{feige2010finding} to analyse the performance of \textsc{Low-Degree-Removal} on this subgraph.\\

\begin{remark}\label{rem:non-robust}
	In contrast to the behaviour of Algorithm~\ref{alg:subsampledcounting} noted in Remark~\ref{rem:robust}, it is unlikely that Algorithm~\ref{alg:subsample+filter}, \textsc{Subsample-And-Filter}, is very robust to misspecification of the size of the planted clique. This is because we seem to be using this size crucially in our filtering step. The algorithm needs an estimate of $k$ that has additive error at most $o(\sqrt{n})$.
\end{remark}

\subsection{Impossibility results}
\label{sec:techn-lb}

A sublinear time algorithm can only look at a subset of entries in the adjacency matrix $A_G$ of the graph. For formal definitions of the graph ensembles, problems, conjectures and model of computation used here, see Section~\ref{sec:notation}. Essentially our model is that the input is presented to the algorithm via the adjacency matrix of the graph, and we assume that querying any entry of this matrix takes unit time. Since accessing an entry of the input takes unit time, if an algorithm runs in time $T(n)$, it can observe at most $O(T(n))$ entries of the input adjacency matrix (Remark~\ref{rem:cantseeall}). 
\\
\begin{definition}[Query set of an algorithm]\ \\
	\label{def:query-set}
	Let $\A$ be any algorithm that takes as input $A_G$, the adjacency matrix of a graph $G= (V,E)$. Define $E_{\A} \subset V \times V$ as the set of entries of $A_G$ that $\A$ queries before it terminates. Since $A_G$ is symmetric, we assume for convenience that $E_{\A}$ is symmetric. That is, if $\A$ queries $(i,j)$, it also queries $(j,i)$\footnote{Obviously, since $A_G$ is symmetric, if $\A$ queries $(i,j)$ there is no need to query $(j,i)$. However, doing so can increase the number of queries (and hence the runtime) by at most a factor of $2$. Since it is convenient to assume that the set $E_{\A}$ is symmetric rather than tracking which of the two queries the algorithm made, we simply assume the algorithm queries both options.}.
\end{definition}

The following simple fact about sublinear time algorithms follows immediately from the fact that in our model of computation any query to an entry of the adjacency matrix $A_G$ takes unit time.\\
\begin{remark}
	\label{rem:cantseeall}
	If $\left\{\A_n \right\}$ is an algorithmic family that runs in time $T_{\A}(n)$, then $\lvert E_{\A_n} \rvert = O( T_{\A}(n))$.
\end{remark}

Our lower bounds all use Remark~\ref{rem:cantseeall}. If the algorithm doesn't see ``enough'' entries, it must work without a fair chunk of information about the input. Without this information, what it \textit{does} see is either statistically (in results that follow immediately from \cite{racz2019finding}) or computationally (because of the $\PCC$) not solvable. Remark~\ref{rem:cantseeall} then implies that if an algorithm has a high success probability, it must have a ``large enough'' running time. In this sense, we convert a polynomial vs superpolynomial hardness gap to a fine-grained (in fact sublinear) hardness gap. 

\subsubsection{Information theoretic impossibility result}

The recent work \cite{racz2019finding} that completely characterizes the cost of the planted clique problem in the setting where the cost is measured in terms of the number of queries made by the algorithm to the adjacency matrix of the input graph, and computation is free. Stated in our notation, Theorem 2 in \cite{racz2019finding} says that if an algorithm $\A_n$ has number of queries $\lvert E_{\A_n} \rvert = o\left(\frac{n^2}{k^2} + n \right)$, then $\A_n$ must fail to recover the planted clique with probability tending to $1$. Combining this with Remark~\ref{rem:cantseeall}, we get\\

\begin{proposition}\cite[Theorem 2]{racz2019finding}
	\label{rem:info-theory-lb}
	Let $k(n) \leq n$ and let $\left\{G_n \right\}$ be a sequence of instances of the planted clique recovery problem ${\sf{PC}_R}(n,k(n))$ (Definition~\ref{def:pc-recovery}).
	Any algorithmic family $\left\{\A_n\right\}$ that runs in time $T_{\A}(n) = o\left(\frac{n^2}{(k(n))^2} + n \right)$ must fail to output the correct planted clique with probability at least $1-o(1)$.
\end{proposition}

While this provides a tight lower bound for cliques of size $k = \Omega(n^{\frac{2}{3}})$, it is quite far from our algorithmic upper bound of $\widetilde{O}(n^{\frac{3}{2}})$ for cliques of size $\Theta(\sqrt{n \log n})$ by only providing an $\Omega(n)$ lower bound. However, Theorem 2 in \cite{racz2019finding} also shows that lower bound techniques relying \emph{purely} on analysing the number of queries required will fail to give better lower bounds. This is because there exists an inefficient algorithm making as few as $\widetilde{O}\left(\frac{n^2}{k^2} + n \right)$ queries that can find the planted clique with good probability. Hence we  we resort to using computational hardness assumptions to show stronger lower bounds. The most natural average case computational hardness assumption in this scenario is the \PCC~(Conjecture~\ref{conj:pc}).

An algorithm for recovery can be easily converted into an algorithm for detection simply by picking some large enough random subset of the output of the algorithm (say $3 \log n$) and checking (in time $O(\log^2 n)$) if all the vertices are connected to each other. So we focus on the detection problem when proving the rest of our lower bounds, since this immediately translates into a recovery lower bound. We show that the (non-)existence of really good sublinear time algorithms for planted clique detection in the `easy' regime seems connected to the (non-)existence of very fast (not necessarily sublinear) algorithms for the detection problem in the `hard' regime.

We show this connection using a slightly different notion of a planted clique problem which we call ${\sf{iidPC_D}}(n,\frac{k}{n})$ (Definition~\ref{def:iid-pc-detection}). In this model, each vertex is included in the clique independently with probability $\frac{k}{n}$. We use this as a proxy to study the vanilla planted clique detection problem ${\sf{PC_D}}(n,k)$ (Definition~\ref{def:pc-detection}) in which a uniformly random set of $k$ vertices form the clique. Due to their similarity, impossibility results for one model give evidence for impossibility theorems in the other. For now we simply state the results we obtain, and defer in-depth discussion of the relation to Sections~\ref{sec:algs-work-all-models} and \ref{sec:lb-iidp-to-pc}.

\subsubsection{Lower bounds for detecting planted cliques of size close to information theoretic threshold from sublinear lower bounds for detection at clique size $k = \widetilde{\Theta}(\sqrt{n})$}

Consider the ${\sf{iidPC}_D}(n,\frac{k}{n})$ problem with $k$ just larger than $\sqrt{n \log^2 n}$. Create a subgraph by only retaining the first $\sqrt{n}$ vertices. Then we have a graph of size $\sqrt{n}$ with a planted clique of size slightly more than $2 \log (\sqrt{n}) $, the information theoretic threshold.

Hence if we could solve the detection problem on a graph of size $n$ with a planted clique near the information theoretic threshold in time $O\left(n^{2+2\delta}\right)$ (for any constant $\delta >0$), then we could solve the original problem in time $\widetilde{O}\left(n^{1+\delta} \right) $. A lower bound on the original problem then translates into a lower bound on the problem at the information theoretic threshold. Moreover, a lower bound of the form $\omega(n)$ would imply a non-trivial superlinear lower bound for detecting small cliques. This indicates that a lower bound of the form $\omega(n)$ will require computational hardness assumptions to show.

We prove the following more general reduction, which yields the discussion above by setting $g(n_1) = \log n_1$.\\
\begin{lemma}
\label{thm:backwards-lb}
Let $0 < \delta < \frac{1}{2}$ be some constant. Let $\omega(1) = g(n_1) = o(\sqrt{n_1})$ be some sequence indexed by $n_1$ and define $k_1(n_1) = g(n_1)\sqrt{n_1}$. Suppose that any algorithmic family $\{A_{n_1}\}$ that attempts to solve ${\sf{iidPC_D}}(n_1,\frac{k_1(n_1)}{n_1})$ in time $T_1(n_1) = {O}\left(n_1^{1+\delta} \right)$ has probability of success at most $\frac{1}{2}+o(1)$. Let $k_2(n_2) = g(n_2^2)$. Then any algorithmic family $\{\A_{n_2}\}$ that attempts to solve ${\sf{iidPC_D}}(n_2,\frac{k_2(n_2)}{n_2})$ in time $T_2(n_2) = O(n_2^{2+2\delta})$ has probability of success at most $\frac{1}{2}+ o(1)$.
\end{lemma}
\begin{proof}
Assume towards a contradiction that there exists an algorithmic family $\A_{n_2}$ that runs in time $T_2(n_2) = O(n_2^{2+2\delta})$ and achieves probability of success at least $p_0$ (for some constant $p_0 > \frac{1}{2}$) when solving ${\sf{iidPC_D}}(n_2,\frac{k_2}{n_2})$. Suppose we are given an instance $G$ of ${\sf{iidPC_D}}(n_1,\frac{k_1}{n_1})$ to solve, and consider the following algorithm. Set $n_2 = \sqrt{n_1}$, pick out the first $n_2$ vertices of the graph $G$, and call the induced subgraph (which we don't have to compute, just provide access to) $G'$. Note that we have set $n_2$ so that $\frac{k_1}{n_1} = \frac{k_2}{n_2}$. The definition of ${\sf{iidPC_D}}$ implies that $G'$ is an instance of ${\sf{iidPC_D}}(n_2,\frac{k_1}{n_1})$ hence also an instance of ${\sf{iidPC_D}}(n_2,\frac{k_2}{n_2})$. We can then use $\A_{n_2}$ to solve it in time $T_2(n_2) = O(n_2^{2+2\delta}) = O(n_1^{1+\delta})$ with success probability at least $p_0$. This gives an algorithmic family that runs in time $T_1(n_1) = O(n_1^{1+\delta})$ and solves ${\sf{iidPC_D}}(n_1,\frac{k_1}{n_1})$ with success probability at least $p_0 > \frac{1}{2}$. This provides the desired contradiction and completes the proof.
\end{proof}

\subsubsection{Sublinear time lower bounds for detecting cliques of size $k = \widetilde{\Theta}(\sqrt{n})$ from the Planted Clique Conjecture}

We have seen that strong lower bounds on the detection problem at clique sizes near $k = \widetilde{\Theta}(\sqrt{n})$ imply non-trivial lower bounds for the detection problem at the information theoretic threshold. In this section we show that this connection between sublinear time algorithms for large cliques and polynomial time algorithms for small cliques goes both ways. If the latter is hard (as codified by the \PCC), we provide some evidence that the former is hard too.

However, our results in this direction are weaker than we would like, and only hold hold for a (reasonable) subclass of all algorithms, namely \textit{non-adaptive rectangular} algorithms. For this subclass, we show that our algorithmic upper bounds have essentially optimal running time. Non-adaptive algorithms are those where the algorithm (possibly using randomness) fixes the entries of the adjacency matrix $A_G$ to query before it begins querying $A_G$. Thus the locations of queries do not depend on the entries of $A_G$. Rectangular algorithms are those in which the set of queries form a combinatorial rectangle (modulo symmetry).

For the two definitions that follow, let $\A$ be an algorithm that takes as input the adjacency matrix $A_G$ of a graph $G = (V,E)$, and $E_{\A}$ be the (symmetric) set of queries it makes to $A_G$ as defined in Definition~\ref{def:query-set}.\\
 
\begin{definition}[Non-adaptive algorithm]\ \\
	\label{def:nonadaptive-alg}
	 $\A$ is non-adaptive if the random variable $E_{\A}$ is independent of the random variable $A_G$.\\
\end{definition}

\begin{definition}[Rectangular algorithm]\ \\
	\label{def:rect-alg}
	$\A$ is rectangular if the random variable  $E_{\A}$ is defined using two random disjoint subsets of vertices $I,J \subset V$ with $I \cap J = \emptyset$ as follows.
	$E_{\A} = \left\{(u,v): u \in J , v \in J, u \neq v \right\} \cup \left\{(u,v): u \in I , v \in J \right\} \cup \left\{(u,v): u \in J, v \in I  \right\}$. See Figure~\ref{fig:rect} for an illustrative example.\\
\end{definition}

\begin{figure}[]
	
	\centering
	\includegraphics[scale=0.5]{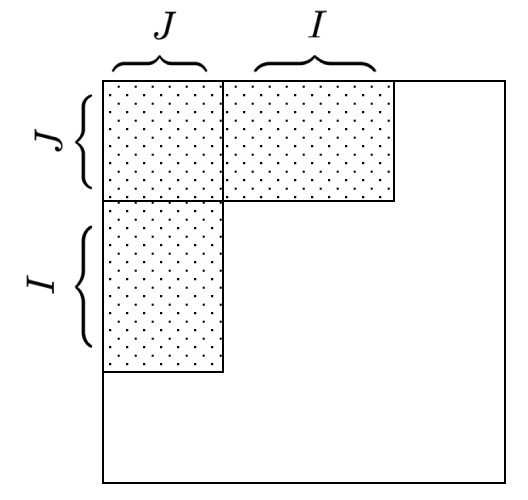}
	\caption{An example of the query set $E_{\A}$ for a rectangular algorithm as defined in Definition~\ref{def:rect-alg}}
	\label{fig:rect}
\end{figure}

\begin{remark}
\label{rem:detec-ub-NA-rect}
While we would like to handle more generic algorithms, restricting our lower bounds to \textit{non-adaptive rectangular} algorithms is not entirely unreasonable. This is because our detection algorithms are non-adaptive and rectangular\footnote{The \textsc{Clique-Completion} subroutine is the only \textit{adaptive} or \textit{non-rectangular} part of Algorithms~\ref{alg:subsampledcounting} ~\ref{alg:subsample+khdac}, or ~\ref{alg:subsample+filter}. Moreover, \textsc{Clique-Completion} is only required for the planted clique recovery problem. A simple tweak to Algorithm~\ref{alg:subsampledcounting} so that it does not use \textsc{Clique-Completion}, but only decides whether or not a planted clique exists based on the largest degree it observes can give a \textit{non-adaptive rectangular} detection algorithm that runs in time $\widetilde{O}(n^{\frac{3}{2}})$. Similarly, removing the \textsc{Clique-Completion} subroutine from Algorithm~\ref{alg:subsample+khdac} while using the modified version of Algorithm~\ref{alg:subsampledcounting} inside it gives a \textit{non-adaptive rectangular} detection algorithm that runs in time $\widetilde{O}(\frac{n^3}{k^3})$ and reliably detects cliques of size $k = \Omega\left(\sqrt{n \log n} \right) $. Similar reasoning also works for Algorithm~\ref{alg:subsample+filter} where we can use a linear time detection algorithm such as edge-counting instead of the \cite{feige2010finding} algorithm.}. Since we are showing lower bounds for the detection version of the problem, our upper bound algorithms do indeed belong to the class of algorithms against which we are showing lower bounds.
\end{remark}

\textbf{Intuition for impossibility result for non-adaptive algorithms:}

Let $\delta>0$ be some constant and consider the detection problem with planted cliques of size $\widetilde{\Theta}(\sqrt{n})$. By a standard argument, we can first use the non-adaptivity of the algorithm to show that we need to only consider algorithms whose queries are deterministic and not randomized. See the proof of Theorem~\ref{thm:n32-lb} for an elaboration of this argument. If the algorithm runs in time $O(n^{\frac{3}{2}-\delta})$ (so even slightly better than our algorithmic upper bound), it can (deterministically) query at most $O(n^{\frac{3}{2}-\delta})$ entries of the adjacency matrix. Under the randomness of the location of the planted vertices, each off-diagonal entry in the adjacency matrix corresponds to a `planted' entry with probability roughly $k^2/n^2 = \widetilde{O}(1/n)$. By linearity of expectation, the expected number of queries the algorithm makes which are `planted' entries is $\widetilde{O}(n^{\frac{1}{2}-\delta})$. This means that we expect the algorithm to obtain evidence of `plantedness' from roughly only $\widetilde{O}(n^{\frac{1}{2}-\delta})$ vertices. According to the $\PCC$, if there are such few planted vertices, it is computationally hard to distinguish between the planted and null models. Thus we might believe that solving the original problem is also computationally hard if we query such few entries. However, we have only been able to leverage this intuition into a formal reduction for rectangular algorithms.\\

\begin{theorem}
	\label{thm:n32-lb}
	Assume the \emph{Planted Clique Conjecture} (Conjecture~\ref{conj:pc}) holds. Let $\delta > 0$ be any constant. Let $k(n)$ be any sequence such that $k(n) = \Omega(\sqrt{n})$. Then if any non-adaptive rectangular algorithmic family $\{A_n\}$ tries to solve ${\sf{iidPC_D}}(n,\frac{k(n)}{n})$ in time $O\left(\frac{n^{3-\delta}}{(k(n))^3}\right)$, it has probability of success at most $\frac{1}{2} + o(1)$.
\end{theorem}
\begin{proof}
The key idea is that the only evidence for a planted clique comes from querying entries for which both vertices are in the planted clique. Fix $n$ and let $k$ denote $k(n)$. Let $\delta>0$ be a constant, and assume towards a contradiction that $\A_n$ is non-adaptive rectangular algorithm that runs in time $O(\frac{n^{3-\delta}}{k^3})$
and solves ${\sf{iidPC_D}}(n,\frac{k}{n})$ with success probability $p_0 > \frac{1}{2}$.
Without loss of generality, we can assume that the query set $E_{\A_n}$, and hence $I$ and $J$, are deterministic\footnote{While the definition of non-adaptive algorithms allows the query set $E_{\A}$ to be randomized, for the sake of proving lower bounds it actually suffices to consider only query sets that are deterministically fixed before the input is provided. This is because the probability of success of the algorithm is an expectation of the success probabilities under each of the random choices of $E_{\A}$ (this is where we use the fact that the algorithm is non-adaptive). This means that there exists at least one choice $E_{\A}$ which achieves the probability of success that the randomized algorithm is guaranteed to achieve. This choice, which can be pre-computed, provides a deterministic algorithm that does at least as well as the randomized non-adaptive algorithm we started with.}. By Remark~\ref{rem:cantseeall}, the number of queries the algorithm makes can not be too big: $\lvert E_{\A_n} \rvert  = O(\frac{n^{3-\delta}}{k^3})$. Hence, $(\lvert I \rvert + \lvert J \rvert)  \lvert J \rvert  = O(\lvert E_{\A_n} \rvert) =  O(\frac{n^{3-\delta}}{k^3})$.
	
We will consider two cases. First, consider the simpler case, where $\lvert J \rvert = O(\frac{n^{1-\frac{\delta}{2}}}{k})$. Since $J$ is so small, with high probability (at least $1-o(1)$) over the randomness of the clique vertices, there are no planted vertices in $J$, that is, $J \cap K = \emptyset$. In this scenario, the distribution of the entries of $E_{\A_n}$, which is what the algorithm sees, will be identical under both $G(n,\frac{1}{2})$ and $\hat{G}(n,\frac{1}{2},\frac{k}{n})$ since our query set will have no planted edges. Hence no algorithm will be able to distinguish between these two cases with success probability greater than $\frac{1}{2}+o(1)$. In this case, the conclusion to our theorem immediately follows.

Now consider the case where $\lvert J \rvert = \Omega(\frac{n^{1-\frac{\delta}{2}}}{k})$, so we must have $\lvert I \rvert + \lvert J \rvert = O(\frac{n^{2-\frac{\delta}{2}}}{k^2})$. The idea here is to use $\A_n$ to solve ${\sf{iidPC_D}}(n',\frac{k'}{n'})$ for 
$k'=o(\sqrt{n'})$---an intractable problem according to the \emph{Planted Clique Conjecture} and Remark~\ref{rem:pcc-to-iidpcc}, hence getting a contradiction.
	
Assume without loss of generality $I\cup J = \{1,\dots,n'\}$ and let $k'$ be such that $\frac{k'}{n'} = \frac{k}{n}$. As $n' = O(\frac{n^{2-\frac{\delta}{2}}}{k^2})$ we get that $k' = O(n'^{\frac{1}{2}-\frac{\delta}{4}})$. Consider an instance $G'$ of ${\sf{iidPC_D}}(n',\frac{k'}{n'})$. We claim that running $\A_n$ on $G'$ succeeds with high probability in detecting cliques. To prove this claim, we notice that $\A_n$ can access its input graph $G'$ only	through $E_{\A_n(G')}$. Moreover, if $G'$ is a null instance (no planted clique), then $E_{\A_n(G')} \stackrel{d}{=} E_{\A_n(G)}$ (identically distributed) where $G$ is a null instance of	${\sf{iidPC_D}}(n,\frac{k}{n})$. And if $G'$ has a planted clique, then our choice of $\frac{k'}{n'} = \frac{k}{n}$ implies that $E_{\A_n(G')} \stackrel{d}{=} E_{\A_n(G)}$ where $G$	is an instance of ${\sf{iidPC_D}}(n,\frac{k}{n})$ with planted clique. Overall, our assumption then implies that $\A_n$ succeeds in solving ${\sf{iidPC_D}}(n',\frac{k'}{n'})$ with probability $p_0 > \frac{1}{2}$ which contradicts the $\emph{iid Planted Clique Conjecture}$.
\end{proof}

\section{On the relationship between the different planted clique variants}

\subsection{Our algorithms are robust to how the planted clique problem is formalized}
\label{sec:algs-work-all-models}

The iid variant ${\sf{iidPC}_D}(n,\frac{k}{n})$  (Definition~\ref{def:iid-pc-detection}) is formally different to but morally similar to $\PCD$. Having stated some impossibility results with ${\sf{iidPC}_D}(n,\frac{k}{n})$ rather than $\PCD$, we note here that our algorithms (or minor tweaks thereof) actually work for this different model too. This lends some credence to showing impossibility in these models as proxies for showing impossibility results for $\PCD$ or $\PCR$.\\
\begin{fact}
	\label{fact:iidPC-basically-PC}
	Let $k(n) \leq n$ and $\omega(1) = f(n) = o(\sqrt{n})$ be any sequence. With probability at least $1-o(1)$, an instance of ${\sf{iidPC}_D}(n,\frac{k(n)}{n})$ is an instance of ${\sf{PC}_D}(n,k'(n))$ for some sequence $k'(n)$ satisfying $\lvert k'(n) - k(n) \rvert \leq f(k(n)) \sqrt{k(n)}$.\\
\end{fact}

\begin{remark}
	\label{rem:algs-work-iid}
	It can be verified that all our algorithms (even \textsc{Subsample-And-Filter}, as stated in Remark~\ref{rem:non-robust}) work as long as the estimate of $k$ they take in as input is within an additive $o(\sqrt{n})$ from the size of the true planted clique. Combining this with Fact~\ref{fact:iidPC-basically-PC}, this means that our algorithms solve ${\sf{iidPC}_D}(n,\frac{k}{n})$ with the same runtime as $\PCD$ and a mildly worse success probability.
\end{remark}

Even if we are hesitant to think of impossibility results about ${\sf{iidPC}_D}(n,\frac{k}{n})$ as being proxies for impossibility results about $\PCD$, we can think of ${\sf{iidPC}_D}(n,\frac{k}{n})$ as being the fundamental problem for which this work describes sublinear time algorithms as well as hardness results.

\subsection{What do these lower bounds formally imply for $\PCD$? }
\label{sec:lb-iidp-to-pc}

Lemma~\ref{thm:backwards-lb} and Theorem~\ref{thm:n32-lb} show that when the planted clique problem is formalized as ${\sf{iidPC_D}}(n,\frac{k}{n})$, the non-existence of very fast sublinear time algorithms for detecting large planted cliques is related to the hardness of detecting small cliques. What does this imply when the problem is formalized using the more vanilla $\PCD$? To discuss this, we note the following two easy lemmas that relate $\PCD$ and ${\sf{iidPC_D}}(n,\frac{k}{n})$.\\
\begin{lemma}[${\sf{iidPC_D}}$ is hard $\rightarrow$ ${\sf{PC_D}}$ is hard]\ \\
	\label{lem:iidPChardthenPChard}
	Let $k(n) \leq n$ and $\omega(1) = f(n) = o(\sqrt{n})$ be some sequences. Suppose that an algorithmic family $\{\A_n\}$ that attempts to solve ${\sf{iidPC_D}}(n,\frac{k(n)}{n})$ has probability of success at most $p_s(n)$. Then there exists a sequence $k'(n)$ (which may depend on $\{A_n\}$) satisfying $\lvert k'(n) - k(n) \rvert \leq f(k(n)) \sqrt{k(n)}$ with probability $1 - o(1)$
	for any $n$, such that if $\{A_n\}$ tries to solve ${\sf{PC_D}}(n,k'(n))$, it has probability of success at most $p_s(n) + o(1)$.\\
\end{lemma}

\begin{lemma}[${\sf{PC_D}}$ is hard $\rightarrow$ ${\sf{iidPC_D}}$ is hard]\ \\
	\label{lem:PChardtheniidPChard}
	Let $k(n) \leq n$ and $\omega(1) = f(n) = o(\sqrt{n})$ be some sequences. Let $k'(n)$ be any sequence satisfying $\lvert k'(n) - k(n) \rvert \leq f(k(n)) \sqrt{k(n)}$ with probability $1 - o(1)$ for any $n$.	
	Suppose that an algorithmic family $\{\A_n\}$ that attempts to solve ${\sf{PC_D}}(n,k'(n))$ has probability of success at most $p_s(n)$. Then if $\{A_n\}$ tries to solve ${\sf{iidPC_D}}(n,\frac{k(n)}{n})$, it has probability of success at most $p_s(n) + o(1)$.
\end{lemma}

\begin{proof}
We prove Lemmas~\ref{lem:iidPChardthenPChard} and ~\ref{lem:PChardtheniidPChard}.
Fix $n$ and let $k,\hat{k}$ denote $k(n),\hat{k}(n)$. Let $G$ be an instance of ${\sf{iidPC_D}}(n,\frac{k}{n})$ 
and let the random variable $\hat{k}$ denote the clique size. Clearly the problem instance $G$ conditioned on the value of $\hat{k}$ is an instance of ${\sf{PC_D}}(n,\hat{k})$. Let $S$ denote the event that an algorithm $\A_n$ succeeds on an instance on $G$, and $E$ denote the event that $\lvert \hat{k}-k \rvert \leq f(k) \sqrt{k}$. Note that $\Prob\left(E^c \right) = o(1)$ because of how the clique is chosen. Then \[\Prob\left(S \right) = \Prob\left(S|E \right)\Prob\left(E \right) +\Prob\left(S|E^c \right)\Prob\left(E^c \right) = \Prob\left(S|E \right)\Prob\left(E \right) + o(1).\] Let $S_{k'}$ denote the event that $\A_n$ succeeds on an instance of ${\sf{PC_D}}(n,k')$, and $k_E$ be the distribution of $\hat{k}$ conditioned on the event $E$ occurring. 

For Lemma~\ref{lem:iidPChardthenPChard} we have $\Prob(S) \leq p_s(n)$ by assumption, and this implies $\mathbb{E}_{\hat{k} \sim k_E}\left[ \Prob\left(S_{\hat{k}} \right) \right] = \Prob\left(S|E \right) \leq p_s(n)+o(1)$ after rearrangement. This implies that for some $k'$ such that $\lvert k'-k \rvert \leq f(k) \sqrt{k}$, $\Prob\left(S_{k'} \right) \leq p_s(n)+o(1)$, which completes the proof.

For Lemma~\ref{lem:PChardtheniidPChard} we have assumed $\Prob\left(S_{k'} \right) \leq p_s(n)$ for all $k'$ such that $\lvert k'-k \rvert \leq f(k) \sqrt{k}$. This means that $\Prob\left(S|E \right) = \mathbb{E}_{\hat{k} \sim k_E}\left[ \Prob\left(S_{\hat{k}} \right) \right] \leq p_s(n)$, which implies $\Prob(S) \leq p_s(n) +o(1)$ and completes the proof. 
\end{proof}

\begin{remark}
	\label{rem:pcc-to-iidpcc}
	It follows immediately from Lemma~\ref{lem:PChardtheniidPChard} that the \emph{Planted Clique Conjecture} (Conjecture~\ref{conj:pc}) implies the \emph{iid Planted Clique Conjecture} (Conjecture~\ref{conj:iidpc}).
\end{remark}

This let us use either the $\PCC$ or the \emph{iid Planted Clique Conjecture} to show impossibility results for ${\sf{iidPC}_D}(n,\frac{k}{n})$ in Theorem~\ref{thm:n32-lb}.

At first glance, this seems great. We can use these lemmas to get analogues of Lemma~\ref{thm:backwards-lb} and Theorem~\ref{thm:n32-lb} for $\PCD$. However, this does not quite work. We illustrate this by trying to show that the $\PCC$ implies the non-existence of rectangular non-adaptive algorithms that can solve $\PCD$ for clique sizes $k = \widetilde{\Theta}(\sqrt{n})$ and run in time $O(n^{\frac{3}{2}-\delta})$ for some constant $\delta > 0$. We already have, from Theorem~\ref{thm:n32-lb} that this fact holds for ${\sf{iidPC_D}}(n,\frac{k}{n})$. If we try to use this with Lemma~\ref{lem:iidPChardthenPChard}, all we can say is that for any algorithmic family $\A_n$, there is \textit{some} sequence $k'$, which is very close to $k$, which this algorithmic family can not solve. However, this need not be the same $k'$ for every algorithm. In effect, it is possible that for every sequence $k'$, there \textit{is} some algorithmic family that can solve it. Thus we can not rule out a fast algorithm for even single such sequence $k'$.

However, this does not mean there is nothing useful we can say. If an algorithm designer (who believes in the $\PCC$) wants to build a non-adaptive rectangular algorithm that can solve ${\sf{PC_D}}(n,k = \sqrt{n \log n})$, we can tell them that their algorithm must \textit{crucially} utilize a very good estimate of the size of the planted clique. This is because their algorithm definitely must fail for some sequence of planted clique sizes that is very close to the true size in the problem instance. As we note in Remark~\ref{rem:algs-work-iid}, the algorithms developed in this work do not crucially utilize such a fine estimate of $k$.

\section{Notation and Technical Definitions}
\label{sec:notation}

We will use standard big $O$ notation ($O, \Theta, \Omega$) and will denote
$\widetilde{O}(f(n))$ to denote ${\sf{poly}}(\log{n})O(f(n))$ and define
$\widetilde{\Theta}, \widetilde{\Omega}$ similarly. Denote
the set of graphs on $n$ vertices by $\mathcal{G}_n$.  For a vertex $v$ in graph
$G = (V, E)$, we will denote its degree by $\deg(v)$. An edge between nodes $u,v
\in V$ is denoted $(u,v)$. We let $\Binnhalf$ denote a Binomial random variable with parameters $\left(n,\frac{1}{2}\right)$. Similarly, {\sf{Bern}}($p$)
denotes a Bernoulli random variable that is $1$ with probability $p$ and $0$
otherwise. Unless stated otherwise, all logarithms are taken base $2$. $[n]$
denotes the set $\{1,2,...,n\}$.

For Definitions~\ref{def:er},~\ref{def:pc} and \ref{def:iidpc}, let $G = (V,E)$ be a graph with vertex set $V$ of size $n$. 
Each definition specifies a subset of vertices $K \subset V$. For all distinct pairs of vertices $u,v \in K$, we add the edge $(u,v)$ to $E$. 
For all remaining distinct pairs of vertices $u,v$, 
we add the edge $(u,v)$ to $E$ independently with probability $\frac{1}{2}$.\\

\begin{definition}[Erd\H{o}s-R\'enyi graph distribution: $\G(n, \frac{1}{2})$]\ \\
	\label{def:er}
	$K$ is the empty set. This distribution on graphs is denoted $\G(n, \frac{1}{2})$.\\
\end{definition}

\begin{definition}[Planted Clique graph distribution: $\G(n, \frac{1}{2},k)$]\ \\
	\label{def:pc}
	Let $K \subset V$ be a set of size $k$ chosen uniformly at random from all ${n \choose k}$ subsets of size $k$. This distribution on graphs is denoted $\G(n, \frac{1}{2},k)$.\\
\end{definition}

\begin{definition}[iid Planted Clique graph distribution: $\hat{\G}(n, \frac{1}{2},p)$]\ \\
	\label{def:iidpc}
	$K \subset V$ is a set such that every vertex $v \in V$ is included in $K$ iid with probability $p$. This distribution on graphs is denoted $\hat{\G}(n, \frac{1}{2},p)$.\\
\end{definition}

\begin{definition}[Planted Clique Detection Problem: ${\sf{PC_D}}(n,k)$]\label{def:pc-detection}\ \\
	This is the following hypothesis testing problem.
	\begin{align}
	\label{eq:hypotheses}
	{\sf H}_0: G \sim \G(n, \frac{1}{2}) \hspace{0.1cm} \text{ and } \hspace{0.1cm}
	{\sf H}_1: G\sim
	\G(n, \frac{1}{2},k)
	\end{align}
\end{definition}

\begin{definition}[Planted Clique Recovery Problem: ${\sf{PC_R}}(n,k)$]\ \\
	\label{def:pc-recovery}
	Given an instance of $G\sim
	\G(n, \frac{1}{2},k)$, recover the planted clique $K$.\\
\end{definition}

\begin{definition}[iid Planted Clique Detection Problem: ${\sf{iidPC_D}}(n,p)$]\ \\
	\label{def:iid-pc-detection}
	This is the following hypothesis testing problem.
	\begin{align}
	{\sf H}_0: G \sim \G(n, \frac{1}{2}) \hspace{0.1cm} \text{ and } \hspace{0.1cm}
	{\sf H}_1: G\sim
	\hat{\G}(n, \frac{1}{2},p)
	\end{align}
\end{definition}

For Conjectures~\ref{conj:pc} and \ref{conj:iidpc}, suppose that $\{\A_n\}$ is a sequence of randomized polynomial time algorithms that take as input the adjacency matrix $A_G$ of a graph $G$ on $n$ vertices,
$\A_n: A_G \rightarrow \{0,1\}$ and let $k(n)$ be a sequence of positive
integers such that $k(n) = O(n^{\frac{1}{2}-\delta})$ for any constant $\delta>0$.\\

\begin{conjecture}[Planted Clique Conjecture]\ \\
	\label{conj:pc}
	 If $G_n$ is a sequence of instances of ${\sf{PC_D}}(n,k(n))$, it holds that
	\[
	\pr_{{\sf H}_0}\left\{\A_n(A_{G_n}) = 0 \right\} +
	\pr_{{\sf H}_1}\left\{\A_n(A_{G_n}) = 1\right\} \leq 1+o(1)   .\]
\end{conjecture}

\begin{conjecture}[iid Planted Clique Conjecture]\ \\
	\label{conj:iidpc}
	If $G_n$ is a sequence of instances of ${\sf{iidPC_D}}(n,\frac{k(n)}{n})$, it holds that
	\[
	\pr_{{\sf H}_0}\left\{\A_n(A_{G_n}) = 0 \right\} +
	\pr_{{\sf H}_1}\left\{\A_n(A_{G_n}) = 1\right\} \leq 1+o(1).\]
\end{conjecture}

To show impossibility results, we also define what it means for an ``algorithm'' to ``solve'' the planted clique problem. Let ${\sf{P}}(n,k(n))$ denote any of the following computational problems - ${\sf{PC}_D}(n,k(n)), {\sf{PC}_R}(n,k(n)), {\sf{iidPC}_D}(n,\frac{k(n)}{n})$.\\

\begin{definition}[`Solving' a problem]\ \\
	\label{def:solving-pc}
	Let $k(n)$, $T(n)$, and $\frac{1}{p(n)}$ be some functions of $n$ such that $k(n) \leq n$. A parametrized family of algorithms $\left\{\A_n \right\}$ is said to run in time $T(n)$ and solve ${\sf{P}}(n,k(n))$ with failure probability at most $p(n)$ if the following happens. For all large enough $n$, when given an instance of ${\sf{P}}(n,k(n))$ as input, $\A_n$ terminates in time $T(n)$ and returns the correct answer with probability at least $1-p(n)$.
\end{definition}

\subsection{Model of Computation}
\label{sec:model}
To talk about sublinear algorithms, it is necessary to specify the model of computation within which we are working. Since we are working with dense graphs (both $\G(n, \frac{1}{2})$ and $\G(n, \frac{1}{2},k)$ have $O(n^2)$ edges with high probability), it is reasonable to assume that the graph is provided via its adjacency matrix. Formally, the algorithm has access to the adjacency matrix $A_G$ of the graph $G$ which is a matrix whose rows and columns are indexed by the vertex set $V$ and entries are defined as follows. $A_G(u,v) = A_G(v,u) = 1$ if $(u,v) \in E$ and $0$ otherwise. Also, $A_G(u,u) = 0$. This is essentially the same as the Dense Graph Model that has been widely studied in the graph property testing literature (see, eg, \cite{goldreich2010introduction}). Computationally, we assume that the algorithm can query any entry of this matrix in unit time. We also assume that sampling a vertex uniformly at random takes unit time, and any other similar edge or vertex manipulation operations take unit time.

%% file: algorithms.tex
We encourage the reader to read the intuition for these algorithms in Section~\ref{sec:techn-algs} before reading these proofs. We simply provide technical details in this section.

\subsection{Clique Completion}
\label{subsec:cliquecompletion}

\begin{algorithm}[]
	\SetAlgoLined
	\KwIn{Graph $G = (V,E) \sim \G(n, \frac{1}{2},k)$, known clique set $S_C \subset V$ }
	\KwOut{Clique $K$}
	Initialize $S = S_C$

	\For{$v \in V\setminus S_C$}{
		\If{$(v,u) \in E$  for all $u \in S_C$}{
			\hspace{25em}\smash{$\left.\rule{0pt}{3.0\baselineskip}\right\}\ \mbox{find common neighbours}$}
			
			Update $S \leftarrow S \cup \{v\}$
		}
	}
	Let $V' \leftarrow S$
	
	Pick (u.a.r) a subset of size $(1+c) \log n$ from $V'$ and call it $S'_C$
	
	Initialize $S' = S'_C$
	
	\For{$v \in V'\setminus S'_C$}{
	\If{$(v,u) \in E$  for all $u \in S'_C$}{
		\hspace{25em}\smash{$\left.\rule{0pt}{3.8\baselineskip}\right\}\ \mbox{post-processing}$}
		
		Update $S' \leftarrow S' \cup \{v\}$
	}
}

	\KwRet $S'$ 
	\caption{\textsc{Clique-Completion}}
	\label{alg:cliquecompletion}
\end{algorithm}
\vspace{.2cm}

\begin{lemma}[Runtime]\label{lem:clique-completion-runtime}
	For any constant $c>0$, \textsc{Clique-Completion} runs in time $O(n \cdot (|S_C| + \log n))$.
\end{lemma}

\begin{lemma}[Correctness]\label{lem:clique-completion-correctness}
	Draw a graph $G$ according to $G \sim \G(n,\frac{1}{2},k)$ and let the set $S_C \subset K$ (the planted clique in the instance $G$) with $\lvert S_C \rvert = (1+c) \log n$ for some constant
    $c>0$. If $k = \omega\left( \log^2 n \right) $ then the output of Algorithm~\ref{alg:cliquecompletion},
    \textsc{Clique-Completion}($G,S_C$) is $K$ with probability at least
    $1-4\frac{(3+c)}{c} \max\left(\frac{(1+c)\log^2 n}{k},  \frac{\log
    n}{n^{\frac{1+c}{3}}}\right)$.
\end{lemma}
\begin{proof}
The algorithm has three stages and
our proof upper bounds the probability of failure of each stage (conditioned on the the previous stages succeeding).

\textbf{Step 1} The first stage
of the algorithm begins with our known clique set, and appends to it every vertex
which is a common neighbor.  We need the number of
non-clique vertices added to not be too big.  Let $A_1$ be the event that 
$\lvert S \setminus K \rvert < \frac{1}{c}(1+(2+c)\log k) := \ell_0$, then Lemma~\ref{lem:false=pos} shows through a simple union bound argument that $\Prob\left(A_1^c \right)
\leq 
\left( \frac{1}{n}\right) ^{\log k}.$

\textbf{Step 2} The algorithm then takes the output set of the first
stage, $S$, and keeps a uniformly random subset $S_C'$ of size
$(1+c)\log{n}$.  let $A_2$ be the event that $S_C' \subset
K$. We show
that $\Prob(A_2^c \mid A_1) \leq \frac{\ell_0 (1 + c) \log{n}}{k}$. To this end, let $b := \lvert V'  \setminus K \rvert = \lvert
S  \setminus K \rvert$ and notice that $b < \ell_0$.  Now,
\begin{align*}
\Prob\left(A_2 \mid A_1 \right) &= \frac{{k \choose (1+c) \log n}}{{k+b \choose (1+c) \log n}}  = \frac{(k-(1+c)\log n + b)(k-(1+c)\log n + b-1)...(k-(1+c)\log n + 1)}{(k + b)(k+ b-1)...(k + 1)} \\
& \geq \left(\frac{k-(1+c)\log n}{k} \right)^b \geq 1- \frac{b(1+c)\log n}{k} \geq 1- \frac{\ell_0 (1+c)\log n}{k}.
\end{align*}

Finally, we analyze the last stage of the algorithm. We notice that $S_C'
\subset K$ implies $K \subseteq S$; that is, if the input to the last stage of
the algorithm is entirely contained within the clique, then the output of the
algorithm contains the clique.  We then prove that the output of the algorithm
is exactly the clique by showing that it has no intersection with the non-clique
vertices.  Let $A_3$ be the event that $S' \setminus K =
\emptyset$.

\textbf{Step 3} We show that 
\[\Prob(A_3^c \mid A_1, A_2) \leq
4n\exp\left(\frac{-k}{54} \right) + \ell_0 n^{-\frac{(1+c)}{3}}.
\]
We need to control the number of clique vertices
any non-clique vertex is connected to. This is done in Lemma~\ref{lem:2kby3}, which gives $\Prob(A_4^c) \leq n\exp\left(\frac{-k}{54} \right)$ where $A_4$ is the event that every non-clique vertex is connected to at
most $2k/3$ clique vertices. Note that
\[
    \Prob(A_3^c \mid A_1, A_2) \leq \Prob(A_3^c \mid A_1, A_2, A_4) + \Prob(A_4^c
    \mid A_1, A_2) \leq \Prob(A_3^c \mid A_1, A_2, A_4) + 4\Prob(A_4^c),
\]
where the last inequality follows as long as $n, k$ are large enough to satisfy $\Prob(A_1) \geq
1/2$ and $\Prob(A_2 \mid A_1) \geq 1/2$. The upshot of choosing $S_C'$ randomly the way we do is that conditioned on $A_1$ and $A_2$, it is a uniformly random subset of the planted clique $K$. Further conditioning on $A_4$,
for a given non-clique vertex in $V'$, the probability that it is connected to all vertices in $S'_C$ is at most \[\frac{{\frac{2k}{3} \choose (1+c)\log n}}{{k \choose (1+c) \log n}} \leq \left(1-\frac{(1+c) \log n}{k} \right)^{\frac{k}{3}} \leq \exp\left(-\frac{(1+c) \log n}{3} \right)  = n^{-\frac{(1+c)}{3}}\]
Union bounding over all the at most $\ell_0$ non-clique vertices in $S'_C$, no non-clique vertex gets added to $S'$  except with probability at most $\ell_0n^{-\frac{(1+c)}{3}}$, which means $\Prob(A_3^c \mid A_1, A_2, A_4) \leq \ell_0n^{-\frac{(1+c)}{3}}$

$A_2 \cap A_3$ is the
event that the clique $K$ is contained in the output $S'$ and that no non-clique
vertex is contained in $S'$.  Thus, $A_2 \cap A_3$ is the success event.  Notice
that $1 - \Prob(A_2 \cap A_3) \leq \Prob(A_1^c)
+ \Prob(A_2^c \mid A_1) + \Prob(A_3^c \mid A_2, A_1)$. Thus, Steps 1-3 imply 
\begin{align*}
    1 - \Prob(A_2 \cap A_3) &\leq \left( \frac{1}{n}\right)^{\log k} + \ell_0
    \left(\frac{(1+c)\log n}{k} +
    n^{-\frac{(1+c)}{3}} \right)  + 4n\exp\left(\frac{-k}{54} \right)  \\
    &\leq 4\frac{(3+c)}{c} \max\left(\frac{(1+c)\log^2 n}{k},  \frac{\log
    n}{n^{\frac{1+c}{3}}}\right). 
\end{align*}
\end{proof}

\subsection{An $\widetilde{O}(n^{3/2})$ algorithm for
finding cliques of size $k = \Theta(\sqrt{n \log n})$}
\label{subsec:subsampledcounting}

\begin{theorem}
\label{thm:subsampledcounting}
Let $8\sqrt{n \log n} \leq k$, and let $L_{in}$ be a user defined parameter. If $  \frac{4n \cdot (\log n)^2}{k} \leq L_{in}$, when given an instance $G$ of $\G(n, \frac{1}{2},k)$, $\textsc{Keep-High-Degree-And-Complete}\left(G,L_{in}\right)$ (Algorithm~\ref{alg:subsampledcounting}) runs in time ${O}(nL_{in}+n \log n)$ and outputs the hidden clique $K$ with probability at least $1- O\left(\frac{\log^2 n}{\sqrt{n}} \right) $. 
\end{theorem}

\begin{proof}
	\textbf{Runtime analysis:}
	Using Lemma~\ref{lem:clique-completion-runtime}, the algorithm clearly runs in time $O(nL_{in} + n \log n)$.
	
	\textbf{Correctness analysis:} 
    We first show that for this size of the clique $k$, degree counting is
    sufficient to separate clique vertices from non-clique vertices.
    We then show that randomly sampling $L_{in}$ vertices
    will yield a subset of the clique of size at least $2\log{n}$.  We conclude
    by invoking Lemma~\ref{lem:clique-completion-correctness} to show that \textsc{Clique-Completion} works correctly with high probability. 

    Let
    $D_{\min}$ denote the event that the minimum degree of a clique vertex is at
    least $d_{\min} = \frac{n}{2} + \frac{k}{2} - \sqrt{3 n \log n}$ and $D_{\max}$ denote
    the event that the maximum degree of a non-clique vertex is at most $d_{\max} = \frac{n}{2} +
    \sqrt{3 n \log n}$.  Note that the degree of a non-clique vertex is
    $\Binnhalf$, so by a Chernoff bound~\ref{lem:chernoff} and a union bound,
    $\Prob(D_{\max}^{c}) \leq \frac{(n-k)}{n^2}  \leq \frac{1}{n}$.  Likewise,
    the degree of a non-clique vertex is $\Binnkhalf + k$, so a similar argument
    shows $\Prob(D_{\min}^c) \leq \frac{k}{n^2} \leq \frac{1}{n}$. 	
    Because $8 \sqrt{n \log n} \leq k$ and 
    setting $T_d = \frac{n}{2} + 2\sqrt{n \log n}$, we have that $d_{\max} < T_d < d_{\min}$. 
    Therefore except with probability at most
    $\frac{1}{n} + \frac{1}{n} \leq \frac{2}{n}$, the degree of all clique nodes
    is larger than $T_d$ and that of all non-clique nodes is smaller than $T_d$.

    Now, we show that randomly sampling $L_{in}$ vertices will
    yield at least $2\log{n}$ clique vertices. Let this random sample of $L_{in}$ vertices be denoted $S_L$. If we divide the clique vertices $K$ into $2 \log n$
    disjoint sets $K_1,K_2,...,K_{2 \log n}$ of equal size $\frac{k}{2 \log
    n}$\footnote{We have omitted certain floors and ceilings for the sake of
readability}, then with high probability we will get at least one vertex from each $K_i$. This implies that we will have at least $2 \log n$ distinct clique vertices in $S_L$. Let $\mathcal{E}_i$ be the event that $S_L \cap K_i = \emptyset$. 
\[
    \Prob\left(\mathcal{E}_i\right) = \left(1-\frac{k}{2n \log n}\right)^{L_{in}}
\leq \exp{\left( \frac{-kL_{in}}{2n \log n}\right) } \leq 2^{\left(
\frac{-kL_{in}}{2n \log n}\right)}. 
\]
Let $\mathcal{E} = \bigcap_{i} \mathcal{E}_i^c$; that is $\mathcal{E}$ is
the event that each $K_i$ has non-empty intersection with $S_L$.  Then, a union
bound shows that $\Prob(\mathcal{E}^{c}) \leq (2 \log n)  2^{\left( \frac{-kL}{2n \log n}\right)}$. Since $L_{in} \geq \frac{4n(\log n)^2}{k}$, this probability of failure is at most $\frac{2 \log n}{n^2}$.

The probability Algorithm~\ref{alg:subsampledcounting} fails can be denoted by $\Prob(C^c)$ where $C$ is the event that \textsc{Clique-Completion} outputs $K$, the planted clique. So we can upper bound $\Prob(C^c) \leq \Prob(C^c , D_{\max},D_{\min},\mathcal{E}) + \Prob(D_{\max}^c) +  \Prob(D_{\min}^c) +  \Prob(\mathcal{E}^c)$. Using our estimates from above, we can upper bound $\Prob(D_{\max}^c) + \Prob(D_{\min}^c) +  \Prob(\mathcal{E}^c) = O(\frac{1}{n}) + O(\frac{\log n}{n^2}) = O(\frac{\log^2 n}{\sqrt{n}})$. Hence it only remains to show that $\Prob(C^c , D_{\max},D_{\min},\mathcal{E}) = O(\frac{\log^2 n}{\sqrt{n}})$ to complete the proof.

Consider a genie who gets the same input as Algorithm~\ref{alg:subsampledcounting} and also knows the location of the planted clique. The genie observes our algorithm, and if \underbar{$S_C$} is not a subset of the planted clique $K$, the genie selects any other set of $2 \log n$ true clique vertices and runs \textsc{Clique-Completion} using this new genie-aided input set instead. We can denote the event that the genie's version of \textsc{Clique-Completion} succeeds as $C_{\mathsf{genie}}$ and by Lemma~\ref{lem:clique-completion-correctness}, we can conclude that $\Prob(C_{\mathsf{genie}}^c) = O(\frac{\log^2 n}{\sqrt{n}})$. This is because the genie-aided algorithm takes as input a graph $G \sim \G(n, \frac{1}{2},k)$ and a true clique subset, which are precisely the conditions on Lemma~\ref{lem:clique-completion-correctness}.

To relate this to our quantity of interest, we note that when $S_{\mathsf{good}}  := D_{\max} \cap D_{\min} \cap \mathcal{E}$ happens, the input \underbar{$S_C$} used by Algorithm~\ref{alg:subsampledcounting} is a subset of $K$ and so the genie-aided algorithm and Algorithm~\ref{alg:subsampledcounting} behave identically conditioned on $S_{\mathsf{good}}$. This means that $\Prob(C^c , S_{\mathsf{good}}) = \Prob(C_{\mathsf{genie}}^c , S_{\mathsf{good}}) \leq \Prob(C_{\mathsf{genie}}^c) = O(\frac{\log^2 n}{\sqrt{n}})$ which completes the proof.
\end{proof}

\begin{algorithm}
\SetAlgoLined
\KwIn{Graph $G = (V,E) = \G(n, \frac{1}{2},k)$, number of vertices to sample $L_{in}$}
\KwOut{Clique $K$}
Initialize $S_C = \emptyset$
\SetKwFor{RepTimes}{repeat}{times}{end}

\RepTimes{$L_{in}$}{
    Sample a random vertex $v \in G$ and compute $\deg(v)$

    \If{$\deg(v) \geq \frac{n}{2} + 2\sqrt{n \log n}$}{
    \hspace{30em}\smash{$\left.\rule{0pt}{3.3\baselineskip}\right\}\ \mbox{high degree}$}
    	
        Update $S_C \leftarrow S_C \cup \{v\}$
    }
}

\If{$\lvert S_C \rvert < 2 \log n$}{
	\KwRet \textsc{Declare Failure}
}

Initialize $\underbar{$S_C$} = \emptyset $

\hspace{33em}\smash{$\left.\rule{0pt}{2.7\baselineskip}\right\}\ \mbox{complete clique}$}

Select $2 \log n$ vertices from $S_C$ uniformly at random and add them to \underbar{$S_C$}
$S \leftarrow \textsc{Clique-Completion}(G,\underbar{$S_C$})$

\KwRet $S$ 
\caption{\textsc{Keep-High-Degree-And-Complete}}
\label{alg:subsampledcounting}
\end{algorithm}
\vspace{.2cm}

\subsection{An $\widetilde{O}\left(\left(n/k\right)^3 + n\right)$ algorithm for
	finding cliques of size $k = \Omega(\sqrt{n \log n})$}
\label{subsec:n3k3}

\begin{theorem}
\label{thm:subsample+khdac}
Let $32 \sqrt{n \log n} \leq k \leq n	$ and set $p = \frac{512 \cdot  n \log n}{k^2}$. Given an instance $G$ of $\G(n, \frac{1}{2},k)$, $\textsc{Subsample-And-KHDAC}\left(G,k,p\right)$ (Algorithm~\ref{alg:subsample+khdac}) runs in time $O\left(\frac{n^3}{k^3} \cdot \log^3 n  + n \log n \right)$ 
and outputs the hidden clique $K$ with probability at least $1-O\left(\frac{\log^2 (pk)}{\sqrt{pk}} \right)$.
\end{theorem}

\begin{proof}
	\textbf{Runtime Analysis:} Since we can sample a random vertex in unit time in our model of computation, sampling $pn$ vertices takes time $O(pn)$. Further, using the running times from Theorem~\ref{thm:subsampledcounting} and Lemma~\ref{lem:clique-completion-runtime}, it is easy to observe that the algorithm runs in time \[O\left(pn + n'L_{in}' + n \log n \right)  = O\left(\frac{pn^2}{k} \cdot \log^2 n   + n \log n \right) = O\left(\frac{n^3}{k^3} \cdot \log^3 n + n \log n \right)\]
	
	\textbf{Correctness Analysis:} 
    We first see that the subsampling step behaves as expected, in the sense that $k_p = \lvert S_P \cap K \rvert$ is roughly equal to $pk$. 
    Let $A_1$ be the event that $\frac{pk}{2} \leq k_p \leq \frac{3pk}{2}$.
    As $k_p$ is a hypergeometric random variable, we use bounds on the concentration of a hypergeometric random variable around its mean (see, for eg, \cite[Theorem 1]{hush2005concentration}) to get that 
    $\Prob(A_1^c) \le 2 \exp\left(-\frac{pk^2}{4n} \right) \leq \frac{2}{n^{128}}$.

	It is easy to observe that $G' \sim \G(n',\frac{1}{2},k_p)$. Now we show that $S_C$ , the output of the \textsc{Keep-High-Degree-And-Complete} subroutine is  equal to $S_P \cap K$ with high probability.
	Denote this event by $A_2$.
	Since $ k' = \frac{pk}{2} \leq  k_p$ and (using that $p \leq \frac{1}{2}$) \[8\sqrt{\lvert S_P \rvert \log {\lvert S_P \rvert}} \leq 8\sqrt{2} \sqrt{pn \log n} = \frac{pk}{2} \leq k_p ,\] the conditions of Theorem~\ref{thm:subsampledcounting} are satisfied if $A_1$ holds, therefore
	$\Prob(A_2^c \mid A_1) = O\left(\frac{\log^2 n'}{\sqrt{n'}} \right)  =  O\left(\frac{\log^2 (pk)}{\sqrt{pk}} \right)$.
	
	To finish the proof, we need to prove that \textsc{Clique-Completion} succeeds
	with high probability. Let $A_3$ denote the probability that the output
	of clique completion (and of the algorithm) $S = K$. Overall, we can then upper bound the probability of failure of Algorithm~\ref{alg:subsample+khdac} as $\Prob(A_3^c) \leq \Prob(A_3^c,A_1,A_2) + \Prob(A_2^c | A_1) + \Prob(A_1^c)$. We have shown that $\Prob(A_2^c | A_1) + \Prob(A_1^c) = O(\frac{1}{n^{128}}) +  O\left(\frac{\log^2 (pk)}{\sqrt{pk}} \right) = O\left(\frac{\log^2 (pk)}{\sqrt{pk}} \right)$. Thus it only remains to show that $\Prob(A_3^c,A_1,A_2) = O\left(\frac{\log^2 (pk)}{\sqrt{pk}} \right)$. We can now define $S_{\mathsf{good}} : = A_1 \cap A_2$ and use the same genie-aided analysis as in the proof of Theorem~\ref{thm:subsampledcounting} to conclude that $\Prob(A_3^c,A_1,A_2) = O\left(\frac{\log^2 (n)}{\sqrt{n}} \right) = O\left(\frac{\log^2 (pk)}{\sqrt{pk}} \right)$.
\end{proof}

\vspace{.2cm}
\begin{algorithm}
	\SetAlgoLined
	\KwIn{Graph $G = (V,E) = \G(n, \frac{1}{2},k)$, clique size $k$, subsampling fraction $p$}
	\KwOut{Clique $K$}
	Set $n' = np$ and $k' = \frac{pk}{2}$
	
	Initialize $S_P = \emptyset$

	Pick $n'$ vertices uniformly at random from $V$ and add them to $S_P$
	\hspace{3.75em}\smash{$\left.\rule{0pt}{1.5\baselineskip}\right\}\ \mbox{subsample}$}
	
	Let $G'$ be the subgraph of $G$ induced by $S_P$
	
	$S_C \leftarrow \textsc{Keep-High-Degree-And-Complete}\left(G',L'_{in}=\frac{4n' \cdot (\log n')^2 }{k'} \right)$
	\hspace{2.75em}\smash{$\left.\rule{0pt}{1.0\baselineskip}\right\}\ \mbox{high degree}$}
	
	\If{$\lvert S_C \rvert < 2 \log n$}{
		\KwRet $k$ vertices chosen uniformly at random from $V$
	}

	Initialize $\underbar{$S_C$} = \emptyset $
	
	\hspace{33em}\smash{$\left.\rule{0pt}{2.7\baselineskip}\right\}\ \mbox{complete clique}$}
	
	Select $2 \log n$ vertices from $S_C$ uniformly at random and add them to \underbar{$S_C$}
	$S \leftarrow \textsc{Clique-Completion}(G,\underbar{$S_C$})$
	
	\KwRet $S$ 
	\caption{\textsc{Subsample-And-KHDAC}}
	\label{alg:subsample+khdac}
\end{algorithm}
\vspace{.2cm}

\subsection{An ${O}\left(n^2/\left(\frac{k^2}{n}\exp{\left(\frac{k^2}{24n}\right)}\right)\right)$ algorithm for finding cliques of size $\omega(\sqrt{n}) = k = o(\sqrt{n \log n})$}
\label{sec:alg-small-k}

\vspace{.2cm}
\begin{algorithm}[ht]
	\SetAlgoLined
	\KwIn{Graph $G = (V,E) = \G(n, \frac{1}{2},k)$, clique size $k$, subsampling fraction $p$ }
	\KwOut{Clique $K$}
	Let $V_1,V_2$ be two disjoint subsets of $V$ of size $\frac{n}{2}$ each.
	
	Initialize $S_P = \emptyset$
	
	\For{$v \in V_1$}{
		With probability $p$, update $S_P \leftarrow S_P \cup \{v\}$
		
		\hspace{31.5em}\smash{$\left.\rule{0pt}{1.8\baselineskip}\right\}\ \mbox{subsample}$}
	}
	\If{$\lvert S_P \rvert > pn$}{
		\KwRet $k$ vertices chosen uniformly at random from $V$
	}
	
	Initialize $S_F = \emptyset$
	
	Set $T_l = \frac{n+k}{4} - 2\sqrt{n}$ and $T_d = \frac{n+k}{4} + 2\sqrt{n}$
	
	\For{$v \in S_P$}{
		\hspace{31.5em}\smash{$\left.\rule{0pt}{3.7\baselineskip}\right\}\ \mbox{filter}$}
		
		\If{$T_l \leq \sum\limits_{u \in V_2} \mathbb{1}_{((u,v) \in E)} \leq T_d$}{$S_F \leftarrow S_F \cup \{v\}$}
		
	}
	
	Set $n' = \lvert S_F \rvert$
	
	Let $G'$ be the subgraph of $G$ induced by $S_F$	\hspace{13.5em}\smash{$\left.\rule{0pt}{1.8\baselineskip}\right\}\ \mbox{\cite{feige2010finding}}$}

	$S_C \leftarrow \textsc{Low-Degree-Removal}\left(G'\right)$ \cite[Theorem 1]{feige2010finding}
	
	\If{$\lvert S_C \rvert < 2 \log n$}{
		\KwRet $k$ vertices chosen uniformly at random from $V$
	}
	
	Initialize $\underbar{$S_C$} = \emptyset $
	
	\hspace{33em}\smash{$\left.\rule{0pt}{2.7\baselineskip}\right\}\ \mbox{complete clique}$}
	
	Select $2 \log n$ vertices from $S_C$ uniformly at random and add them to \underbar{$S_C$}
	$S \leftarrow \textsc{Clique-Completion}(G,\underbar{$S_C$})$
	
	\KwRet $S$ 
	\caption{\textsc{Subsample-And-Filter}}
	\label{alg:subsample+filter}
\end{algorithm}
\vspace{.2cm}

\begin{theorem}\label{thm:subsample+filter}
	Let $\omega\left(\sqrt{n} \right)  = k = o\left(\sqrt{n \log n}\right)$, and let $G$ be an instance of $\G\left(n,\frac{1}{2},k\right)$. Set $p = \frac{n}{k^2}\exp\left(\frac{-k^2}{24n} \right)$. Then $\textsc{Subsample-And-Filter}\left(G,k,p \right)$ (Algorithm~\ref{alg:subsample+filter}) runs in time \[O(pn^2) = {O}\left(n^2/\left(\frac{k^2}{n}\exp{\left(\frac{k^2}{24n}\right)}\right)\right) = o(n^2) \] and outputs the planted clique except with probability at most $\frac{1}{3}+o(1)$.
\end{theorem}

\begin{proof}
	Note that since $k = o(\sqrt{n \log n})$, we have that $p = \omega(n^{-\epsilon})$ for any constant $\epsilon>0$. So we have $pk = \omega(n^{0.49})$ and $pn = \omega(n^{0.99})$.

	\textbf{Runtime Analysis:} 
	
	The subsampling step takes time $O(n)$. If $\lvert S_P \rvert >  pn$, the algorithm terminates with a further $O(k)$ runtime. This would mean a runtime bounded by $O(n)$ which is also in $O(pn^2)$ .
	
	If, on the other hand, $\lvert S_P \rvert \leq pn$, then we need to compute $\deg(v)$ for at most $pn$ vertices and each such computation takes time at most $O(n)$. This step thus takes time $O(pn^2)$. By using runtime bounds from \cite{feige2010finding} and Lemma~\ref{lem:clique-completion-runtime}, susbequent steps of the algorithm take time $O(p^2n^2 + n \log n)$ which is $O(pn^2)$. Hence the complete algorithm has a runtime that is $O(pn^2)$.
	
	\textbf{Correctness Analysis:} We assume the notation set in the algorithm. We analyze each stage of the algorithm.

	\textbf{Step 1} First, we show that subsampling steps behave as expected. Let $k_1$ and $k_2$ be random variables denoting the number of planted vertices in $V_1$ and $V_2$ respectively. We must have $k_1 + k_2 = k$. Let $P_0$ denote the event that $k/2 - \sqrt{n} \leq k_1 \leq k/2 + \sqrt{n}$. $P_0$ implies that 
	$k/2 - \sqrt{n} \leq k_2 \leq k/2 + \sqrt{n}$. We show that the probability $\Prob(P_0^c)$ is small and so then assume for the rest of the proof that $P_0$ holds. Since $k_1$ is a hypergeometric random variable, using concentration bounds from \cite[Theorem 1]{hush2005concentration} we have $\Prob(P_0^c) \leq 2\exp(-n/k)$.
	
	Now controlling the subsampling step that is used to obtain the set $S_P$, define $P_1$ to be the event that $\lvert
	\lvert S_P \rvert - 0.5pn \rvert \leq 0.25pn$ and $P_2$ denote the event that
	$\lvert \lvert S_P \cap K \rvert - pk_1 \rvert \leq 0.5pk_1$. Using Chernoff bounds from Lemma~\ref{lem:subsampling-conc}, we have $\Prob(P_1^c|P_0) \leq 2\exp(\frac{-pn}{24})$ and $\Prob(P_2^c|P_0) \leq 2\exp(\frac{-pk_1}{12}) \leq 2\exp(\frac{-pk}{48})$. Defining the event $P := P_0 \cap P_1 \cap P_2$, we can upper bound $\Prob(P^c) \leq \Prob(P_1^c|P_0) + \Prob(P_2^c|P_0) + \Prob(P_0^c) = O(\exp(-n/k))$. For brevity, we let $\hat{n} = \lvert S_P \rvert$ and $\hat{k} = \lvert S_P \cap K \rvert$.

	\textbf{Step 2} We now assume the event $P$ happens and aim to show that the filtering step also behaves as expected and analyse the size of $S_F$ and $S_F \cap K$. The subgraph induced by $S_F$ is the input to the \textsc{Low-Degree-Removal} subroutine of \cite{feige2010finding} so we want to show that $|S_F \cap K|$ is relatively large and that $|S_F|$ is not too large, so that the subroutine works as expected.
	To this end, we define the
	event $F_1$ to denote the event that $S_F$ does not contain too many
	non-clique vertices.  That is, we let $F_1 = \left\{\lvert S_F \setminus
	(S_F \cap K) \rvert \leq pn \exp\left(\frac{-k^2}{24n} \right)\right\}$.
	Similarly, we let $F_2$ define the event that $S_F \cap K$ is fairly large:
	$F_2 = \left\{\lvert \lvert S_F \cap K \rvert - p_0\hat{k} \rvert \leq
	0.5p_0\hat{k}\right\}$ (for some parameter $p_0$ to be defined later). 
	
	If $v \in S_P \setminus (S_P \cap K)$ (that is, it is not a clique vertex), we upper bound the probability that it will be added to  $S_F$. Using a Chernoff bound (Lemma~\ref{lem:chernoff}) and $\omega(\sqrt{n}) = k$
	\begin{align*}
	& \Prob\left(v \in S_F | P , v \in S_P \setminus (S_P \cap K)\right)  = \Prob\left( \lvert \sum\limits_{u \in V_2} \mathbb{1}_{((u,v) \in E)} - \frac{n+k}{4} \rvert \leq 2\sqrt{n} \right) \\
	& \leq \Prob\left(  \sum\limits_{u \in V_2} \mathbb{1}_{((u,v) \in E)} \geq \frac{n+k}{4}  - 2\sqrt{n} \right) \leq \exp\left(\frac{-k^2}{12n} \cdot \left(1-\frac{8\sqrt{n}}{k} \right)^2 \right)
	\end{align*}
	By linearity of expectation, \[\mathbb{E}\left[\lvert S_F \setminus (S_F \cap K) \rvert | P \right] \leq (\hat{n}-\hat{k}) \Prob\left(v \in S_F | P , v \in S_P \setminus (S_P \cap K)\right) = O\left(pn \exp\left(\frac{-k^2}{12n} \cdot \left(1-\frac{8\sqrt{n}}{k} \right)^2 \right) \right) .\] Using Markov's inequality, \[ \Prob(F_1^c| P) = \Prob \left(\lvert S_F \setminus (S_F \cap K) \rvert \geq pn \exp\left(\frac{-k^2}{24n} \right) | P \right) =  O\left(\exp\left(\frac{-k^2}{24n} + \frac{4k}{3\sqrt{n}}\right) \right) = O\left(\exp\left(\frac{-k^2}{48n}\right) \right).\]
	
	We now analyse $S_F \cap K$ and show that it is relatively large (conditioned on $P$). For any $v \in S_P \cap K$, using Chernoff (Lemma~\ref{lem:chernoff})
	\begin{align*}
	1-p_0 & := \Prob\left(v \notin S_F |P \right) = \Prob\left( \lvert \sum\limits_{u \in V_2} \mathbb{1}_{((u,v) \in E)}-\frac{n+k}{4}\rvert \geq 2\sqrt{n} \right) \\ & \leq \Prob\left( \sum\limits_{u \in V_2} \mathbb{1}_{((u,v) \in E)} \geq \frac{n+k}{4}+2\sqrt{n} \right) + \Prob\left( \sum\limits_{u \in V_2} \mathbb{1}_{((u,v) \in E)} \leq \frac{n+k}{4}-2\sqrt{n} \right) \\ & = \Prob\left( \sum\limits_{u \in V_2} \mathbb{1}_{((u,v) \in E)} - k_2 \geq \frac{n-2k_2}{4}+\frac{k-2k_2}{4}+2\sqrt{n} \right) \\& + \Prob\left( \sum\limits_{u \in V_2} \mathbb{1}_{((u,v) \in E)} - k_2 \leq \frac{n-2k_2}{4}+\frac{k-2k_2}{4}-2\sqrt{n} \right) \\& \leq 2\Prob\left( \sum\limits_{u \in V_2 \setminus K} \mathbb{1}_{((u,v) \in E)} - \mathbb{E}\left[\sum\limits_{u \in V_2 \setminus K} \mathbb{1}_{((u,v) \in E)}\right] \geq 1.5\sqrt{n} \right) \leq 2\exp(-3) \leq 1/2.
	\end{align*}
	This gives us $p_0 = \Prob\left(v \in S_F | P \right) \geq \frac{1}{2}$. Since $\lvert S_F \cap K \rvert$ is a sum of $\hat{k}$ independent {\sf{Bern}($p_0$)} random variables, by Lemma~\ref{lem:subsampling-conc} \[\Prob\left(\lvert \lvert S_F \cap K \rvert - p_0\hat{k} \rvert \leq 0.5p_0\hat{k} |P \right)  \geq 1- 2\exp\left(\frac{-p_0\hat{k}}{12} \right) \geq 1- 2\exp\left(\frac{-pk_1}{48} \right) \geq 1- 2\exp\left(\frac{-pk}{192} \right).\] 
	
	Denoting $F := F_1 \cap F_2$, we have thus upper bounded $\Prob(F^c | P) \leq \Prob(F_1^c | P) + \Prob(F_2^c | P) = O\left(\exp\left(\frac{-k^2}{48n}\right) \right)$.

	\textbf{Step 3} We now analyze the event $A_1$ that
	$\textsc{Low-Degree-Removal}$ \cite{feige2010finding} succeeds conditioned on $P,F$. Conditioned on $P \cap F$, we can observe that the subgraph $G'$ induced by the vertex set $S_F$ is distributed as $\G(|S_F|,\frac{1}{2},|S_F \cap K|)$. This is because we have not yet used the randomness from any of the edges in this subgraph. Note that even though we do not know $|S_F\cap K|$, this is not a problem because $\textsc{Low-Degree-Removal}$ does not need the clique size as an input.
	
	Moreover, we can see that $\lvert S_F \rvert \leq pn \exp\left(\frac{-k^2}{24n}\right) + \frac{3p_0\hat{k}}{2} =  O\left(pn \exp\left(\frac{-k^2}{24n}\right) \right) $ and $\lvert S_F \cap K \rvert = \omega(pk)$. This gives \[\sqrt{\lvert S_F \rvert} =
	O\left(\sqrt{pn \exp\left(\frac{-k^2}{24n}\right)} \right)  =
	O\left(pk \right)  = o\left(\lvert S_F \cap K \rvert \right),\] which means that the conditions required for $\textsc{Low-Degree-Removal}$ to work are satisified and we can upper bound the failure probability as $\Prob(A_1^c | P,F) \leq \frac{1}{3}$ \cite[Theorem 1]{feige2010finding}.

	\textbf{Step 4} Finally, we can analyze the event Algorithm~\ref{alg:subsample+filter} proceeds to the
	\textsc{Clique-Completion} step and succeeds. Let $A_2$ denote the event that the output of \textsc{Clique-Completion} is the planted clique $K$. Then the failure probability of the algorithm is $\Prob(A_2^c) \leq \Prob(A_2^c , A_1,F,P) + \Prob(A_1^c | P,F) + \Prob(F^c | P) + \Prob(P^c)$. We have already shown that $\Prob(A_1^c | P,F) + \Prob(F^c | P) + \Prob(P^c) \leq \frac{1}{3}+o(1)$ and so it we only need to show that $\Prob(A_2^c , A_1,F,P) = o(1)$ to complete the proof. Define $S_{\mathsf{good}} :=  A_1 \cap F \cap P$ and note that conditioned on $S_{\mathsf{good}}$, the input vertex set \underbar{$S_C$} to \textsc{Clique-Completion} is a subset of the planted clique $K$ because $2\log n = O\left(pk \right) = O\left(\lvert S_F \cap K \rvert \right)$. Hence we can use the same genie-aided analysis technique as in the proof of Theorem~\ref{thm:subsampledcounting} to show that $\Prob(A_2^c | A_1,F,P) = O\left(\frac{\log^2 n}{n^{0.49}} \right) = o(1)$. This completes the proof.
\end{proof}

%% file: auxlemmas.tex
We state the Chernoff bound we use here, for the convenience of the reader.
\begin{lemma}\label{lem:chernoff}
	Let $X = \sum\limits_{i=1}^{n} X_i$ where $X_i$ are independent {\sf{Bern}}($p_i$) random variables. Let $\mu = \sum\limits_{i=1}^{n} p_i$, and $\delta \in (0,1)$. Then
	\[\Prob\left( X \geq (1+\delta) \mu \right)  \leq \exp\left(\frac{-\mu \delta^2}{3} \right) \text{ and } \Prob\left( X \leq (1-\delta) \mu \right)  \leq \exp\left(\frac{-\mu \delta^2}{3} \right)  \]
\end{lemma}

This implies the following subsampling concentration lemma that proves useful.
\begin{lemma}\label{lem:subsampling-conc}
	Let $V$ be a set of size $n$, and let $K \subset V$ be of size $k$. Let $S_P$ be a subset of $V$ formed by including every element with probability $p$, and excluded otherwise. Then
	\[\Prob \left(0.5pn \leq \lvert S_P \rvert \leq 1.5pn \right) \geq 1- 2\exp\left(\frac{-pn}{12} \right) \text{ and }\Prob \left(0.5pk \leq \lvert S_P \cap K \rvert \leq 1.5pk \right) \geq 1- 2\exp\left(\frac{-pk}{12} \right)  \]
\end{lemma}

With high probability, \textit{any} clique subset of size $\Theta(\log n)$ has at most $\Theta(\log n)$ non-clique vertices connected to every vertex of the subset. The analysis of such a lemma is contained in the proof of~\cite[Lemma
2.9]{dekel2014finding}.
\begin{lemma}\label{lem:false=pos}
Let $G = (V,E) = \G(n, \frac{1}{2},k)$ and $S$ be any arbitrary subset of the planted clique $K$ with $|S| = (1+c) \log n$ for some constant $c >0$. Let $T$ be the set of all non-clique vertices that are connected to every vertex in $S$. Then, except with probability at most $\left(
\frac{1}{n}\right) ^{\log k}$, $|T| \leq
\frac{1}{c}(1+(2+c)\log k)$
\end{lemma}
\begin{proof}
  Fix $S \subset K$ such that $\lvert S \rvert = (1 + c)\log{n}$.
The probability there exists a subset of non-clique vertices of size
$\ell$ connected to every element in $S$ is at most $\binom{n}{\ell}2^{-\ell
	(1+c)\log{n}}$.  A union bound then implies that the probability there exists a
subset of non-clique vertices of size at least $\ell_0 =
\frac{1}{c}(1+(2+c)\log k)$ connected to every
element in $S$ is at most $\sum_{\ell = \ell_0}^{n - k} \binom{n}{\ell}2^{-\ell
	(1+c)\log{n}} \leq 2^{-(c\ell_0 - 1)\log{n}}$.  Further union bounding over all
subsets of $K$ of size $(1 + c)\log{n}$ implies the probability of our desired event not happening is at most
\[
\binom{k}{(1 + c)\log{n}}2^{-(c\ell_0 - 1)\log{n}} \leq
2^{(1+c)\log n \log k}2^{-(cl_0-1)\log n} = 2^{- \log k \log n} = \left(
\frac{1}{n}\right) ^{\log k}.
\]
\end{proof}

We now control the number of clique vertices
any non-clique vertex is connected to.
\begin{lemma}\label{lem:2kby3}
Let $G = (V,E) = \G(n, \frac{1}{2},k)$, and let $d$ be the maximum number of clique vertices connected to a non-clique vertex. Then $\Prob(d \geq \frac{2k}{3}) \leq n \exp\left(\frac{-k}{54}\right)$
\end{lemma}
\begin{proof}
A Chernoff bound
(Lemma~\ref{lem:chernoff}) shows that any given non-clique vertex is connected to more than $\frac{2k}{3}$ clique vertices with probability at most $\exp\left(\frac{-k}{54}\right)$. A union bound over the at most $n$ non-clique vertices finishes the proof.
\end{proof}